\documentclass[journal]{IEEEtran}
\ifCLASSINFOpdf
\else
\fi
\usepackage{subfigure}

\usepackage{cite}
\usepackage{graphicx,graphics}
\usepackage{epstopdf}
\DeclareGraphicsExtensions{.eps}
\graphicspath{{figs/}}
\usepackage{url}
\usepackage{bm}

\usepackage{amsmath,amssymb,amsthm}  
\usepackage{paralist}

\usepackage{algorithm}
\usepackage{algorithmic}
\usepackage{multirow}
\renewcommand{\algorithmicensure}{\textbf{Iteration:}}
\renewcommand{\algorithmiclastcon}{\textbf{Output:}} 

\newcommand{\e}{\begin{equation}}
\newcommand{\ee}{\end{equation}}
\newcommand{\en}{\begin{equation*}}
\newcommand{\een}{\end{equation*}}
\newcommand{\eqn}{\begin{eqnarray}}
\newcommand{\eeqn}{\end{eqnarray}}
\newcommand{\bmat}{\begin{bmatrix}}
\newcommand{\emat}{\end{bmatrix}}
\newcommand{\BIT}{\begin{itemize}}
\newcommand{\EIT}{\end{itemize}}
\newcommand{\eg}{{$e.g.,$}}
\newcommand{\ie}{{$i.e.,$}}
\newcommand{\jj}{\mathrm{j}}

\newcommand{\diag}{\mathop{\bf diag}}

\newcommand{\dist}{\mathop{\bf dist{}}}

\newcommand{\vb}{\mathbf b}

\newcommand{\vd}{\mathbf d}
\newcommand{\ve}{\mathbf e}
\newcommand{\vf}{\mathbf{f}}

\newcommand{\vh}{\mathbf h}

\newcommand{\vk}{\mathbf k}

\newcommand{\vp}{\mathbf p}

\newcommand{\vr}{\mathbf r}
\newcommand{\vs}{\mathbf s}
\newcommand{\vt}{\mathbf t}
\newcommand{\vu}{\mathbf{u}}
\newcommand{\vv}{\mathbf v}
\newcommand{\vw}{\mathbf w}
\newcommand{\vx}{\mathbf x}
\newcommand{\vy}{\mathbf y}
\newcommand{\vz}{\mathbf z}

\newcommand{\vtheta}{\bm \theta}

\newcommand{\mA}{\mathbf A}

\newcommand{\mC}{\mathbf C}
\newcommand{\mD}{\mathbf D}

\newcommand{\mG}{\mathbf G}
\newcommand{\mH}{\mathbf H}
\newcommand{\mI}{\mathbf I}
\newcommand{\mJ}{\mathbf J}
\newcommand{\mK}{\mathbf K}

\newcommand{\mP}{\mathbf P}

\newcommand{\mR}{\mathbf R}
\newcommand{\mS}{\mathbf{S}}

\makeatletter
\DeclareOldFontCommand{\rm}{\normalfont\rmfamily}{\mathrm}
\DeclareOldFontCommand{\sf}{\normalfont\sffamily}{\mathsf}
\DeclareOldFontCommand{\tt}{\normalfont\ttfamily}{\mathtt}
\DeclareOldFontCommand{\bf}{\normalfont\bfseries}{\mathbf}
\DeclareOldFontCommand{\it}{\normalfont\itshape}{\mathit}
\DeclareOldFontCommand{\sl}{\normalfont\slshape}{\@nomath\sl}
\DeclareOldFontCommand{\sc}{\normalfont\scshape}{\@nomath\sc}
\makeatother

\newcounter{oursection}

\usepackage{hyperref} 
\usepackage{cleveref} 
\Crefname{figure}{Fig.}{Figs.}

\usepackage{mathtools}
\usepackage{xcolor} 

\definecolor{goldenrod}{rgb}{1,0.87,0.26}

\newtheorem{prop}{Proposition}

\usepackage[export]{adjustbox}
\usepackage{xcolor}
\usepackage{epstopdf}
\definecolor{darkred}{rgb}{0.6,0,0}
\definecolor{darkgreen}{rgb}{0,0.5,0}
\definecolor{darkblue}{rgb}{0,0,0.5}
\usepackage{soul}
\usepackage{caption}
\usepackage{tikz,pgfplots}
\pgfplotsset{compat=1.5.1}
\usepgfplotslibrary{groupplots}
\usetikzlibrary{intersections,calc,arrows,matrix,spy,math,calc,fit,positioning, decorations.pathmorphing}
\tikzset{snake it/.style={decorate, decoration=snake}}

\usetikzlibrary{arrows.meta,shapes.geometric}
\newcommand{\tikzcircle}[2][red,fill=red]{\tikz[baseline=-0.5ex]\draw[#1,radius=#2] (0,0) circle ;}%
\DeclareRobustCommand{\tikzarr}[4]{\tikz[baseline=0ex]\draw[-{Latex}] (#1,#2) -- (#3,#4);}

\makeatletter
\newcommand{\mybox}{%
    \collectbox{%
        \setlength{\fboxsep}{1.5pt}%
        \fbox{\BOXCONTENT}%
    }%
}
\makeatother

\begin{document}
%
\title{Diffraction Tomography with Helmholtz Equation: Efficient and Robust Multigrid-Based Solver}
%
%
%
\author{
Tao Hong, Thanh-an Pham, Eran Treister, and Michael Unser 
\thanks{T. Hong and T.-a. Pham contributed equally.}
\thanks{This work was supported in part by the European Research Council
(ERC) under the European Union’s Horizon 2020 research and innovation programme, Grant Agreement 692726 GlobalBioIm: Global integrative framework
for computational bio-imaging and in part by the Israel Science Foundation (grant No. 1589/19).}
\thanks{T. Hong is with the Department of Computer Science, Technion-Israel Institute of Technology, 3200003 Haifa, Israel. (Email: {hongtao@cs.technion.ac.il}).}
\thanks{T.-a. Pham and M. Unser are with Biomedical Imaging Group, \'{E}cole poly- technique f\'{e}d\'{e}rale de Lausanne, 1015 Lausanne, Switzerland (Email: \{thanh- an.pham,michael.unser\}@epfl.ch).}
\thanks{E. Treister is with the Department of
Computer Science at the Ben-Gurion University of the Negev, 8410501 Be'er
Sheva, Israel. (Email: {erant@cs.bgu.ac.il}).}
}

%
%

\markboth{}
{Shell \MakeLowercase{\textit{et al.}}: Bare Demo of IEEEtran.cls for IEEE Journals}
%



\maketitle

\begin{abstract}
Diffraction tomography is a noninvasive technique that estimates the refractive indices of unknown objects and involves an inverse-scattering problem governed by the wave equation.
Recent works have shown the benefit of nonlinear models of wave propagation that account for multiple scattering and reflections.
In particular, the Lippmann-Schwinger~(LiS) model defines an inverse problem to simulate the wave propagation.
Although accurate, this model is hard to solve when the samples are highly contrasted or have a large physical size.
In this work, we introduce instead a Helmholtz-based nonlinear model for inverse scattering.
To solve the corresponding inverse problem, we propose a robust and efficient multigrid-based solver.
Moreover, we show that our method is a suitable alternative to the LiS model,
especially for strongly scattering objects.
Numerical experiments on simulated and real data demonstrate the effectiveness of the Helmholtz model, as well as the efficiency of the proposed multigrid method.
\end{abstract}

%
\IEEEpeerreviewmaketitle

\section{Introduction}
\IEEEPARstart{T}{he} purpose of diffraction tomography~(DT) is to recover the refractive-index~(RI) map of an object in a noninvasive manner~\cite{wolf1969three}.
The sample is probed with a series of tilted incident waves, while the resulting complex-valued scattered waves are recorded for each illumination~\cite{jin2017tomographic}.
From these measurements, one reconstructs the RI map by solving an inverse-scattering problem.
The quality of the reconstruction depends on the angular diversity and the accuracy of the forward imaging model.
When the illumination is a time-harmonic field, the wave propagation through the sample is governed by the Helmholtz equation under the scalar-diffraction theory.
To simplify the reconstruction problem, pioneering works used a linear model to approximate the physical process.
For instance, the Born~\cite{wolf1969three} and Rytov~\cite{devaney1981inverse} approximations are mainly valid for weakly scattering samples.
Recent studies showed that regularization techniques could improve the quality of reconstruction and counteract the presence of noise and the missing-cone problem~\cite{sung2009optical,lim2015comparative,yang2020deep}.
Moreover, 
nonlinear models that are able to properly account for multiple scattering as well as reflections could also improve the quality of reconstruction~\cite{dubois2005retrieval,chaumet2009three,mudry2012electromagnetic,kamilov2015learning,kamilov2016optical,lim2018learning, lim2019high,liu2017seagle,soubies2017efficient,pham2020three,kadu2020high}.

In particular, recent works relied on the Lippmann-Schwinger~(LiS) equation---an integral formulation of the Helmholtz equation---to design an accurate nonlinear forward model~\cite{liu2017seagle,soubies2017efficient,pham2020three}.
A similar model, called the discrete-dipole approximation, can additionally account for polarization~\cite{draine1994discrete,girard2010nanometric,zhang2016far}.
However, when facing particularly strongly scattering samples, we observed that the recent solvers \cite{liu2017seagle,soubies2017efficient,pham2020three} for the LiS model suffer from slow convergence.
Furthermore, the LiS model involves a discrete convolution over a domain larger than the one of interest, which increases the overall computational burden. One may use well-designed preconditioners for the LiS equation \cite{ying2015sparsifying,liu2018sparsify}, but these methods are memory-consuming and need a significant setup time, which hinders their application to inverse-scattering problem.

In this work, we are interested in solving the inverse-scattering problem of strongly scattering objects.
To that end, we introduce a nonlinear imaging model that is based directly on the Helmholtz equation and that relies on a robust and efficient multigrid~(MG) solver.
We show that our method is as accurate as the LiS model while remaining efficient even for strongly scattering samples.
Similarly to the approach developed in \cite{soubies2017efficient}, we also provide an explicit expression of the Jacobian matrix of our model to easily evaluate the gradient in the data-fidelity term.
Our numerical experiments show that the proposed MG solver accurately resolves challenging inverse-scattering problems, while mitigating the prohibitive computational cost of the LiS model.
\subsection{Outline}
The rest of the paper is organized as follows: In \Cref{sec:PhysicalModel}, we introduce the physical model of diffraction tomography and review the LiS model.
In \Cref{sec:MGSolverHelmholtz}, we present the proposed MG-based solver.
In \Cref{sec:RecProblemFormulationandOptimization}, we formulate the inverse-scattering problem and propose an optimization algorithm to solve it.
In \Cref{sec:Exps},
we study the robustness and efficiency of the proposed MG method with numerical experiments on simulated and real data.

\subsection{Notations}
Scalar and continuously defined functions are denoted by italic letter~(\eg{} $\eta_\mathrm{b}\in\mathbb{R}_{\geq 0}$, $f \in L_2$). Vectors and matrices are denoted by bold lowercase and uppercase letters, respectively~(\eg{}~$\vv\in\mathbb{R}^N,\mA\in\mathbb{C}^{N\times N}$); $\|\vv\|$ stands for the $\ell_2$-norm of $\vv\in\mathbb{R}^N$ and $\langle\vv_1,\vv_2\rangle$ denotes the inner products between the vectors $\vv_1,\vv_2\in\mathbb{R}^N$.
The imaginary unit $\jj$ is such that $\jj^2=-1$ and the real part of a complex number~$z$ by $\Re(z)$.
The diagonal matrix~$\diag(\vv)\in\mathbb{R}^{N\times N}$ is formed out of the entries of~$\vv$.
For a matrix~$\mA\in\mathbb{R}^{N\times N}$, $\diag(\mA)\in\mathbb{R}^N$ denotes the diagonal of $\mA$.
The matrix~$\mI_N\in\mathbb{R}^{N\times N}$ is the identity.
%
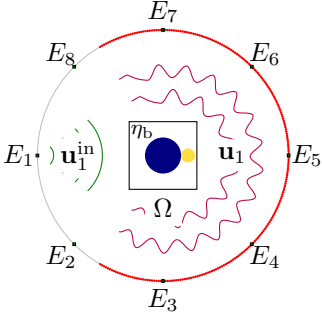
\begin{figure}
\centering
\def\dist{1.67}
\def\emsiz{1pt}
\def\anguin{45}
\def\nwaveuin{4}
\def\percdist{0.1}
\def\angus{120}
\def\distsensors{0.125}
\def\baseshift{0.35}
\newcommand{\view}[1]{E_{#1}}
\newcommand{\sensors}[1]{S_{#1}}
\begin{tikzpicture}[scale = 1]
\draw[lightgray] (0,0) circle (\dist);
\foreach \ang in {-120,-119,...,120} {
\draw[red,fill=red] (\ang:\dist) circle (0.25pt);%
}
\foreach[evaluate = \ang as \i using int(\ang/45 + 5) ] \ang in {-180,-135,...,135} {
\draw[black,fill=darkgreen, shift = {(\ang:\dist)}] ({- \emsiz/2},{- \emsiz/2}) rectangle ({\emsiz/2},{\emsiz/2});
\node[shift = (\ang:7pt)] at (\ang:\dist) {$\view{\i}$};
}
\foreach \pos in {1,2,...,\nwaveuin} {
\draw[darkgreen] ({-\dist+\percdist*\dist*\pos},{-\percdist*\dist*\pos*sin(\anguin)}) arc[start angle=-\anguin,end angle=\anguin, radius = \percdist*\dist*\pos]; 
rad}
\node[fill=white,circle,inner sep=0pt] at (-0.675*\dist,0) {$\vu^\mathrm{in}_1$};
\draw[purple,snake it] ({0.5*\dist*cos(-\angus)},{0.5*\dist*sin(-\angus)}) arc[start angle = {-\angus},end angle = {\angus}, radius = {0.5*\dist}];
\draw[purple,snake it] ({0.7*\dist*cos(-\angus)},{0.7*\dist*sin(-\angus)}) arc[start angle = {-\angus},end angle = {\angus}, radius = {0.7*\dist}];
\node[fill=white,circle,inner sep=0pt] at (0.55*\dist,0) {$\vu_1$}; 
\fill [scale = 3,darkblue] (0,0) circle(0.08/1.67*\dist);
\fill [scale = 3,goldenrod] ({(0.08+0.031)/1.67*\dist},0) circle(0.031/1.67*\dist);
\draw[black] (-3*0.15/1.67*\dist, -3*0.15/1.67*\dist) rectangle (3*0.15/1.67*\dist,3*0.15/1.67*\dist);
\node at (-1.83*0.15/1.67*\dist,2.1*0.15/1.67*\dist) {\footnotesize$\eta_\mathrm{b}$};
\node [align=center,anchor=north,fill=white] at (0,-3*0.15/1.67*\dist) {$\Omega$};
\end{tikzpicture}
    \caption{Acquisition setup of diffraction tomography. The sensors (small round on the circle) collect the illumination from~$E_1$. In this example, $8$~views are acquired.}
 \label{fig:ODTRefractive}
\end{figure}

\section{Physical Model} \label{sec:PhysicalModel}
\subsection{Continuous-Domain Formulation}
\label{sec:PhysicalModel:sub:Continuous}
Let us consider an object of RI map $\eta: \Omega \rightarrow \mathbb{R}$ over some spatial domain $\Omega\subset \mathbb{R}^d$ ($d=2,3$).
The object is immersed in a medium of RI $\eta_\mathrm{b}$ and is illuminated by a plane wave (\Cref{fig:ODTRefractive})
\begin{equation}
    u^\mathrm{in}(\vx,t)=\Re\left(u_0\mathrm{e}^{\jj \langle\vk, \vx\rangle-\jj \omega t}\right),
\end{equation}
where~$\vk,\vx\in \mathbb R^d$, $\omega\in\mathbb{R}$,  $u_0\in \mathbb C$, and $t\in\mathbb{R}$ denote the wave vector, the spatial coordinates, the angular pulsation, the complex envelope, and the time, respectively.
Since the incident field is a time-harmonic wave,
the time-independent total field~$u(\vx)$ at location~$\mathbf{x}$ is well described by the inhomogeneous Helmholtz equation~\cite{kamilov2016optical,soubies2017efficient}
\begin{equation}
    \nabla^2u(\vx)+k_0^2\eta^2(\vx) u(\vx)=0\label{eq:HelmTotal},
\end{equation}
where $k_0=\omega/c$ is the wave number in free space and \mbox{$c\approx 3\times 10^8 \mathrm{m}/\mathrm{s}$} the velocity of light.
Denote by \mbox{$u^\mathrm{in}(\vx) = u_0 \mathrm{e}^{\jj\langle\vk, \vx\rangle}$} the incident wave in space and $u^{\mathrm{sc}}(\vx) = (u(\vx) - u^\mathrm{in}(\vx))$ the scattered wave field.
Note that $u^\mathrm{in}(\vx)$ is a solution of the homogeneous Helmholtz equation
\mbox{$\nabla^2 u^\mathrm{in}(\vx)+k_0^2\eta_\mathrm{b}^2u^\mathrm{in}(\vx)=0$}. Then, \eqref{eq:HelmTotal} reads
\begin{equation}
    -\nabla^2u^{\mathrm{sc}}(\vx)-k_0^2\eta^2(\vx)u^{\mathrm{sc}}(\vx)=f(\vx)u^\mathrm{in}(\vx)\label{eq:HelmScatterRightIncident},
\end{equation}
where $f(\vx)=k_0^2(\eta^2(\vx)-\eta_\mathrm{b}^2)$ is the scattering potential function, which is the quantity that we wish to recover.

Equivalently, the integral form of \eqref{eq:HelmTotal} is known as the Lippmann-Schwinger equation and describes the wave propagation~\cite{liu2017seagle,soubies2017efficient,pham2020three} as
\begin{equation}
    u(\vx)=u^\mathrm{in}(\vx)+\int_\Omega g(\vx-\vz)f(\vz)u(\vz)d\vz\label{eq:LSIntegral}.
\end{equation}
With the assumption of Sommerfeld's radiation condition~\cite{sommerfeld1949partial}, the Green's function $g(\vx):\mathbb{R}^d\rightarrow \mathbb{C}$ in \eqref{eq:LSIntegral} is defined as~\cite{schmalz2010derivation}
\begin{equation}\label{eq:green}
g(\vx)=\left\{
\begin{array}{ll}
    \frac{\jj}{4}H_0^{(1)}(k_0\eta_\mathrm{b}\|\vx\|), & d = 2  \\
    \frac{1}{4\pi}\frac{\mathrm{e}^{\jj k_0\eta_\mathrm{b}\|\vx\|}}{\|\vx\|}, & d=3,
\end{array}
\right.
\end{equation}
where $H_0^{(1)}$ is the Hankel function of the first kind.
In DT, the acquisition setup records the (complex-valued) total field at the sensor positions~$\Gamma\subset\mathbb{R}^d$ with~$\Gamma \cap \Omega=\emptyset$.
\subsection{Discrete Forward Model}
\label{sec:PhysicalModel:sub:DiscreteForwMod}
To solve an inverse scattering problem,
recent works propose a two-step forward model based on~\eqref{eq:LSIntegral}~\cite{liu2017seagle,soubies2017efficient}, which we refer to as LiS methods.
The authors first discretize $\Omega$ into $N$ points lying on a uniform grid.
Then, the (nonlinear) forward imaging operator $\mH_\mathrm{LiS}(\vf):\mathbb{R}^N\rightarrow\mathbb{C}^M$ returns the scattered field on the sensor~$\Gamma$ as
\begin{align}
\mH_\mathrm{LiS}:\vf &\mapsto \tilde{\mG}\diag(\vf)\vu_\mathrm{LiS}(\vf), \label{eq:DiscreteForward}
\end{align}
where $\vf\in\mathbb{R}^N$ and $\vu^\mathrm{in}\in\mathbb{C}^N$ are the discrete (and vectorized) counterparts of the scattering potential and the incident field on $\Omega$, respectively.
The matrix~$\tilde{\mG} \in \mathbb{C}^{M\times N}$ encodes the convolution with the Green's function in~\eqref{eq:LSIntegral} in such way that it gets the scattered field on~$\Gamma$.
In~\eqref{eq:DiscreteForward},~$\vu_\mathrm{LiS}(\vf):\mathbb{R}^N\rightarrow\mathbb{C}^N$ is the discrete total field on~$\Omega$ and is computed from~\eqref{eq:LSIntegral}.
\subsection{Computation of \texorpdfstring{$\vu_\mathrm{LiS}(\vf)$}~~from the LiS Equation 
}\label{sec:PhysicalModel:sub:ComputeLiS}
In the LiS methods, $\vu_\mathrm{LiS}(\vf)$ is determined based on the inversion of the discretized form of \eqref{eq:LSIntegral}~\cite{soubies2017efficient,liu2017seagle}
\begin{equation}
\left(\mI_N-\mG\diag(\vf)\right)\vu_\mathrm{LiS}(\vf)  = \vu^\mathrm{in},\label{eq:invertLSm}
\end{equation}
where
$\mG\in\mathbb{C}^{N \times N}$ encodes the convolution with the Green's function in~\eqref{eq:LSIntegral}~\cite{pham2020three}.
In \cite{liu2017seagle, soubies2017efficient}, the normal equation of~\eqref{eq:invertLSm} was iteratively solved via Nesterov accelerated gradient descent~(NAGD) or conjugate gradient~(CG) methods.
In~\cite{pham2020three}, \eqref{eq:invertLSm} was directly solved by the biconjugate-gradient stabilized method (Bi-CGSTAB)~\cite{van1992bi}.
Since Bi-CGSTAB solves \eqref{eq:invertLSm} faster than both NAGD and CG~\cite{pham2020three}, we use Bi-CGSTAB in this work.

The Green's function~\eqref{eq:green} is oscillatory and has a singularity at $\vx = \mathbf{0}$, which is challenging to discretize. 
In~\cite{pham2020three}, the corresponding convolution operator~$\mG$ is properly discretized through a truncation trick~\cite{vainikko2000fast, vico2016fast} and the main computational burden amounts to four fast Fourier transforms (FFT) per iteration of Bi-CGSTAB. In practice, the FFTs are actually applied to a space $2^d$ times larger than the domain of interest so as to approximate an aperiodic convolution.
The LiS methods then require one to store the Fourier transform of the truncated Green's function ($2^dN$ points), which might lead to memory issues when $N$ is large, for instance in three-dimensional problems.

Our numerical experiments show that the LiS method with Bi-CGSTAB still requires a large number of iterations to converge when the object is strongly scattering.
Now, a slow convergence hinders the efficiency of the LiS methods because the total field needs to be computed repeatedly. Based on those observations, we propose instead to solve the Helmholtz equation~\eqref{eq:HelmScatterRightIncident} directly with an efficient and robust MG solver.

\section{Multigrid Methods}
\label{sec:MGSolverHelmholtz:sub:MGMethods}

Let us assume that we want to solve in terms of~$\vu^h\in\mathbb{C}^N$ the system of linear equations
\begin{equation}
\label{eq:genlinprob}
    \mA^h\vu^h=\vb,
\end{equation}
where~$\vb\in\mathbb C^N$ and where~$\mA^h\in\mathbb{C}^{N\times N}$ is a symmetric positive-definite matrix corresponding to the discretization of some partial differential equation with mesh size~$h$.
To solve~$\eqref{eq:genlinprob}$, there exist
local relaxation methods, also called local smoothers (\eg{} Jacobi, Gauss-Seidel, Kaczmarz).
Their behavior was studied in early works~\cite{brandt1977multi}.
These techniques refine a current estimate of the solution in an iterative manner.
Let ${\vv}^{h}_\nu$ be the estimate of the solution~$\vu^h$ of~\eqref{eq:genlinprob} at the $\nu$th iteration.
Then, the Jacobi method with damped factor $\omega_\mathrm{S}\in(0,1]$ sets ${\vv}^{h}_{\nu+1}$ as
\begin{equation}\label{eq:dampedJ}
{\vv}^{h}_{\nu+1} = {\vv}^{h}_\nu-\omega_\mathrm{S} \mD_{\mA^h}^{-1}(\mA^h {\vv}^{h}_\nu-\vb),
\end{equation}
where the diagonal matrix $\mD_{\mA^h}=\diag(\diag(\mA^h))\in\mathbb C^{{N}\times{N}}$ is formed out of the diagonal of~$\mA^h$.
Let the residual of~\eqref{eq:genlinprob} at the $\nu$th iteration be
\begin{equation}
\label{eq:linsyserr}
    \mA^h\ve^h_\nu = \vr^h_\nu,
\end{equation}
where~$\vr^h_\nu = (\vb - \mA^h{\vv}^{h}_\nu) \in \mathbb C^{{N}}$ is the residual and $\ve^h_\nu=(\vu^h-{\vv}_\nu^h)$ the current error.
Evidently, ${\vv}^{h}_\nu$ is the solution of \eqref{eq:HelmScatterRightIncident:discrete} if $\ve^h_\nu=\bm 0$. 
Using the eigenvectors of $\mA^h$ as the basis to represent $\ve^h_\nu$, we refer to the eigenvectors corresponding to the large (small, respectively) eigenvalues as the high-frequency (low-frequency, respectively) components of $\ve^h_\nu$.

The local Fourier analysis (LFA)---a rigorous quantitative analysis tool for MG methods~\cite{brandt1994rigorous}---showed that local relaxation methods can efficiently eliminate the high-frequency components in $\ve^h_\nu$. However, LFA also showed that these methods require more iterations to remove the low-frequency components.
To exploit this specificity, MG methods rely on relaxation steps and coarse-grid correction~(CGC).
A relaxation step typically consists in few iterations of a local smoother with low computational cost~(\eg{} the Jacobi method) so as to efficiently eliminate the high-frequency components in the error~$\ve^h_\nu$.
The CGC then addresses the remaining error (\ie{} the low-frequency components) by solving the related problem on a coarser grid.
The CGC benefits from two aspects: (a) the coarse problem has fewer variables, thus reducing the computational cost; (b) the low-frequency error on the fine problem is usually well approximated on the coarse problem and looks bumpier, which again can be efficiently eliminated by local smoothers~\cite{brandt1977multi}.

We now introduce two operators. The prolongation operator $\mP$ transfers a vector from a coarse grid to a fine grid. The restriction operator $\mR$ transfers a vector from a fine grid to a coarse grid. We refer the reader to \cite{trottenberg2000multigrid,xu2017algebraic} about the choice of $\mP$ and~$\mR$ for different problems.
The choice of $\mP$ and $\mR$ in this work will be specified in \Cref{sec:MGSolverHelmholtz:sub:ProposedMGMethod}.
The typical formulation of MG methods is the two-grid cycle presented in \Cref{alg:TG}. It first calls $\nu_1\in \mathbb N_{\geq 0}$ pre-relaxation step(s), which provides an approximate solution ${\vv}^h$.
Subsequently, one restricts the residual to $\vr^h = (\vb - \mA^h{\vv}^h)$ to $\vr^{2h} = \mR\vr^h$ and solves the coarse problem
\begin{equation}
\label{eq:linsyserrcoarse}
    \mA^{2h}\ve^{2h} = \vr^{2h},
\end{equation}
to obtain the error $\ve^{2h}$. Then, the error $\ve^{2h}$ is prolongated to correct the current estimate~$\vv^h$ on the fine grid.
Finally, $\nu_2\in\mathbb N_{\geq 0}$ post-relaxation steps usually follow and we get the final estimate~$\widehat{\vv}^h$.
The matrix $\mA^{2h}$ can be formulated via a Galerkin formulation (\ie{} $\mA^{2h}=\mR\mA^h\mP$) or via discretizing \eqref{eq:genlinprob} with mesh size $2h$.
\begin{algorithm}[t]
	\caption{Two-grid cycle}
	\label{alg:TG}        
	\begin{algorithmic}[1]
	    \REQUIRE $\mA^h\in \mathbb C^{{N} \times {N}},\vb\in \mathbb C^{N},\vv^h\in \mathbb C^{N}$, and $\nu_1,\nu_2 \in \mathbb N_{\geq 0}$.
		\lastcon $\widehat{\vv}^h \leftarrow TwoGrid(\mA^h,\vb,\vv^h)$.
		\STATE {Call $\nu_1$ times pre-relaxation: $\vv^h\leftarrow Relax(\mA^h,\vb,\vv^h)$.}
		\STATE {Compute the residual $\vr^h \leftarrow \vb-\mA^h \vv^h$.}
        \STATE {Restrict $\vr^h$ for the coarse problem $\vr^{2h} \leftarrow \mR\vr^h$.}
		\STATE {Compute $\ve^{2h}$ by solving $\mA^{2h} \ve^{2h} = \vr^{2h}$.}
		\STATE {Prolong $\ve^{2h}$ and apply CGC: $\vv^h\leftarrow \vv^h + \mP \ve^{2h}$.} 
		\STATE {Apply $\nu_2$ times post-relaxation: $\widehat{\vv}^h = Relax(\mA^h,\vb,\vv^h)$.}
	\end{algorithmic}
\end{algorithm}
In practice, \eqref{eq:linsyserrcoarse} will face the same issue as \eqref{eq:linsyserr} if local smoothers are used.
We can then apply an additional two-grid cycle to solve \eqref{eq:linsyserrcoarse}.
Such a recursive procedure can continue until the coarse problem is solved exactly, which yields a so-called MG algorithm.
In~\Cref{alg:MG}, we present two MG schemes, namely V-cycle with $CycleType=1$ and W-cycle with $CycleType=2$.
Note that the numbers~$\nu_1,\nu_2$ of relaxation steps are not necessarily the same at each level, which allows us to balance the speed of convergence with the cost of computation.
We refer the reader to \cite{briggs2000multigrid,trottenberg2000multigrid} and the references therein for more details about MG methods.
In \Cref{fig:typeCycle}, we display a four-level scheme of the V-cycle and W-cycle to highlight their difference.

\begin{algorithm}[t]
	\caption{Multigrid cycle}            
	\label{alg:MG}                 
	\begin{algorithmic}[1]
	    \REQUIRE $\mA^h\in \mathbb C^{{N} \times {N}},\vb\in \mathbb C^{N},\vv^h\in \mathbb C^{N}$,
	    $CycleType\in\mathbb N_{\geq 1}$,
	    $h>0$,
	    $\nu_1,\nu_2 \in \mathbb N_{\geq 0}$.
		\lastcon
		$\widehat{\vv}^h \leftarrow MGCycle(\mA^h,\vb,\vv^h,CycleType, h)$.
		\IF {coarsest level}
		\STATE Solve (exactly)
		$\mA^{h} \ve^{h} = \vb$.
		\RETURN $\ve^h$.
		\ENDIF
		\STATE {Call $\nu_1$ times pre-relaxation: $\vv^h\leftarrow Relax(\mA^h,\vb,\vv^h)$.}
		\STATE {Compute the residual $\vr^h = \vb-\mA^h \vv^h$.}
		\STATE {Restrict $\vr^h$ for the coarse problem $\vr^{2h} = \mR\vr^h$.}
		\STATE {$CycleCount = 1$.}
		\STATE {$\ve^{2h}\leftarrow \bm 0$.}
		\FOR {$CycleCount\leq CycleType$}
		\STATE {$\ve^{2h} \leftarrow MGCycle(\mathcal \mA^{2h},\vr^{2h},\ve^{2h},CycleType,2h)$.}
		\STATE {$CycleCount\leftarrow CycleCount+1$.}
		\ENDFOR
		\STATE {Prolong $\ve^{2h}$ and apply CGC: $\vv^h\leftarrow \vv^h + \mP\ve^{2h}$.} 
		\STATE {Apply $\nu_2$ times post-relaxation: $\widehat{\vv}^h = Relax(\mathcal \mA^h,\vb,\vv^h)$.}
	\end{algorithmic}
\end{algorithm}

\begin{figure}
\centering
\newcommand{\rC}{0.05cm}
\newcommand{\szmax}{1.8}
\newcommand{\ddist}{\szmax/6}
\subfigure[V-cycle]{
\begin{tikzpicture}
\pgfsetarrowsstart{Latex}
\node[circle,draw,fill=white,minimum size=\rC] (v1) at (0,\szmax) {};
\node[circle,draw, fill=white,minimum size=\rC] (v2) at (\ddist,2/3*\szmax) {};
\node[circle,draw,fill=white, minimum size=\rC] (v3) at (2*\ddist,\szmax/3) {};
\node[circle,draw,fill=black, minimum size=\rC] (v4) at (3*\ddist,0) {};
\node[circle,draw, fill=white,minimum size=\rC] (v5) at (4*\ddist,\szmax/3) {};
\node[circle,draw, fill=white,minimum size=\rC] (v6) at (5*\ddist,2/3*\szmax) {};
\node[circle,draw, fill=white,minimum size=\rC] (v7) at (6*\ddist,\szmax) {};
\draw[-{Latex}] (v1) -- (v2);
\draw[-{Latex}] (v2) -- (v3);
\draw[-{Latex}] (v3) -- (v4);
\draw[-{Latex}] (v4) -- (v5);
\draw[-{Latex}] (v5) -- (v6);
\draw[-{Latex}] (v6) -- (v7);
\end{tikzpicture}
}
\hspace{1cm}
\subfigure[W-cycle]{
\begin{tikzpicture}
\node[circle,draw, fill=white,minimum size=\rC] (v1) at (0,\szmax) {};
\node[circle,draw,fill=white, minimum size=\rC] (v2) at (\ddist,2/3*\szmax) {};
\node[circle,draw, fill=white,minimum size=\rC] (v3) at (2*\ddist,\szmax/3) {};
\node[circle,draw,fill=black, minimum size=\rC] (v4) at (3*\ddist,0) {};
\node[circle,draw,fill=white,minimum size=\rC] (v5) at (4*\ddist,\szmax/3) {};
\node[circle,draw,fill=black, minimum size=\rC] (v6) at (5*\ddist,0) {};
\node[circle,draw, fill=white,minimum size=\rC] (v7) at (6*\ddist,\szmax/3) {};
\node[circle,draw,fill=white, minimum size=\rC] (v8) at (7*\ddist,2/3*\szmax) {};
\node[circle,draw,fill=white, minimum size=\rC] (v9) at (8*\ddist,\szmax/3) {};
\node[circle,draw,fill=black, minimum size=\rC] (v10) at (9*\ddist,0) {};
\node[circle,draw, fill=white,minimum size=\rC] (v11) at (10*\ddist,\szmax/3) {};
\node[circle,draw,fill=black, minimum size=\rC] (v12) at (11*\ddist,0) {};
\node[circle,draw,fill=white, minimum size=\rC] (v13) at (12*\ddist,\szmax/3) {};
\node[circle,draw,fill=white, minimum size=\rC] (v14) at (13*\ddist,2/3*\szmax) {};
\node[circle,draw,fill=white,minimum size=\rC] (v15) at (14*\ddist,\szmax) {};
\draw[-{Latex}] (v1) -- (v2);
\draw[-{Latex}] (v2) -- (v3);
\draw[-{Latex}] (v3) -- (v4);
\draw[-{Latex}] (v4) -- (v5);
\draw[-{Latex}] (v5) -- (v6);
\draw[-{Latex}] (v6) -- (v7);
\draw[-{Latex}] (v7) -- (v8);
\draw[-{Latex}] (v8) -- (v9);
\draw[-{Latex}] (v9) -- (v10);
\draw[-{Latex}] (v10) -- (v11);
\draw[-{Latex}] (v11) -- (v12);
\draw[-{Latex}] (v12) -- (v13);
\draw[-{Latex}] (v13) -- (v14);
\draw[-{Latex}] (v14) -- (v15);
\end{tikzpicture}
}
    \caption{Four-level representation of V-cycle and W-cycle. The symbol~\tikzcircle[fill=white]{3pt} refers to the relaxation procedure; \tikzcircle[fill=black]{3pt} refers to the coarsest level, which is usually solved exactly;
    \tikzarr{0}{0}{0.15}{0.3}
    refers to the prolongation; \tikzarr{0}{0.3}{0.15}{0.}~refers to the restriction.
    }
    \label{fig:typeCycle}
\end{figure}
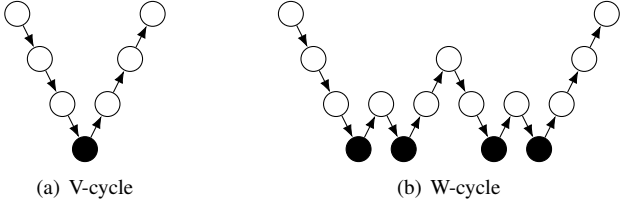

Remarkably, the additional computational cost of such a multilevel approach is low. Here, we take the computational cost of one V-cycle as example.
Let us define the computational cost of one local relaxation on the finest problem as one work-unit (WU) and examine how many WUs are needed for one V-cycle.
In this discussion, we omit the cost of $\mP$ and $\mR$ which amounts to at most $20$\% of the cost of the entire cycle~\cite{briggs2000multigrid}. Moreover, 
we also assume that the computational cost on the coarsest problem is negligible.
When the mesh-size of the coarse problem is doubled~(\ie{} $2h$), the dimension is reduced to $\frac{1}{2^d}$ of the fine grid.
At each level, the computational cost amounts to $2^{-pd}\text{WU}$, where $p=0,1,\ldots,N_\mathrm{level}$ and $N_\mathrm{level}$ denotes the number of levels.
Overall, the computational cost of one V-cycle is
\begin{equation}
(v_1+v_2) \sum_{p=0}^{N_\mathrm{level}} 2^{-pd}  \text{WUs} < \frac{v_1+v_2}{1-2^{-d}}\text{WUs}.
\end{equation}
For $d=2$, we obtain an upper bound $\frac{4(v_1+v_2)}{3}\text{WUs}$ which suggests that, compared with a single level, multilevel does not dramatically increase the computation~\cite{trottenberg2000multigrid}. 

\section{Multigrid-Based Solver for the Helmholtz Model}
\label{sec:MGSolverHelmholtz}
We now discuss the computation of the total field from the Helmholtz equation instead of the LiS equation. Moreover, we present MGH as an MG-based solver for the Helmholtz equation, tailored for inverse scattering.
In particular, MGH efficiently computes the total field for strongly scattering objects.
The computations are carried on a domain only slightly larger than the one of interest, which contrasts with the requirements of the LiS method (\ie{} $2^d$ times larger).
In what follows, we first describe the discretization of \eqref{eq:HelmScatterRightIncident}.
We further discuss the challenges of a plain application of MG methods to the Helmholtz equation and present an heuristic way to address these issues.

\subsection{Discretization of the Helmholtz Equation}
\label{sec:MGSolverHelmholtz:sub:Discretization}
We discretize the Helmholtz equation on a domain of interest $\Omega$ with $d=2$ (\Cref{fig:HelmholtzBoundaryCond}).
To avoid artificial reflections near the boundary, we consider an extended domain~$\Omega_\mathrm{e}$ with an additional absorbing boundary layer (ABL) 
that gradually damps the outgoing waves.
To that end, we multiply $k_0^2\eta^2(\vx)$ in \eqref{eq:HelmScatterRightIncident} with

\begin{equation}
\alpha(\vx) =
    1-\jj\beta \frac{\|\vx - \mathcal{P}_\Omega(\vx)\|^2}{L},
\end{equation}
where $\beta\geq 0$ is an arbitrary parameter,
\mbox{$L > 0$} is the thickness of the ABL,
and $\mathcal{P}_\Omega(\vx)$ is the orthogonal projection of~$\vx$ on~$\Omega$.
Without loss of generality, $\Omega_\mathrm{e}$ is normalized to $[0,1]^2$ and $N_\mathrm{e}$ points are used to discretize \eqref{eq:HelmScatterRightIncident} on $\Omega_\mathrm{e}$.
The points lie on a uniform grid with mesh-size $h = \frac{1}{\sqrt{N_\mathrm{e}}-1}$ at the positions \mbox{$\vx = (mh,nh)$} with $m,n = 0,\ldots, (\sqrt{N_\mathrm{e}}-1)$.

Let $\Omega_\mathrm{e}^h$, $\vu^{\mathrm{sc},h}(\vf^{\,h})\in\mathbb{C}^{{N_\mathrm{e}}}$ and $\vb^h\in\mathbb{C}^{{N_\mathrm{e}}}$ denote the discretized $\Omega_\mathrm{e}$ with mesh-size $h$, the discretized and vectorized versions of~$u^\mathrm{sc}(\vx)$ and $f(\vx)u^\mathrm{in}(\vx)$ on~$\Omega_\mathrm{e}^h$, respectively. The discretization of~\eqref{eq:HelmScatterRightIncident} yields the system of linear equations
\begin{equation}
\mA^h_\mathrm{Hel} \vu^{\mathrm{sc},h}(\vf^{\,h})=\vb^{h},\label{eq:HelmScatterRightIncident:discrete}
\end{equation}
where $\mA^h_\mathrm{Hel} \in \mathbb{C}^{{N_\mathrm{e}} \times {N_\mathrm{e}}}$ is the discretization of \mbox{$\left(-\nabla^2-k_0^2\eta^2(\vx)\right)$} on $\Omega^h_\mathrm{e}$. Specifically, the second-order finite difference is used to discretize the Laplace operator~$\nabla^2$. The ($n\sqrt{N_\mathrm{e}}+m$)th row of \eqref{eq:HelmScatterRightIncident:discrete}
reads
\begin{align}
\frac{-(u^h_{m-1,n}+u^h_{m+1,n}+u^h_{m,n-1}+u^h_{m,n+1})+ 4u^h_{m,n}}{h^2}&\nonumber\\
-(\eta^h_{m,n})^2u^h_{m,n}=b_{m,n}^h,&\label{eq:discretizedHelm:h}
\end{align}%
where $u_{m,n}$ is the $(n\sqrt{N_\mathrm{e}}+m)$th element of $\vu^{\mathrm{sc},h}(\vf^{\,h})$ and  $\mathit{\eta}_{m,n}^h$ denotes the sample $k_0\eta(mh,nh)$.
Moreover, the first-order Sommerfeld radiation condition is used to avoid a nonphysical solution~\cite{elman2001multigrid}.
At the boundary, this translates into
\begin{align}
u_{m,-1}^h &= (1 + \jj h\eta_{m,0})u_{m,0}^h,\nonumber\\
u_{m,\sqrt{N_\mathrm{e}}}^h &= (1 + \jj h\eta_{m,\sqrt{N_\mathrm{e}}-1})u_{m,\sqrt{N_\mathrm{e}}-1}^h,\nonumber\\
u_{-1,n}^h &= (1 + \jj h\eta_{0,n})u_{0,n}^h,\nonumber\\ u_{\sqrt{N_\mathrm{e}},n}^h &= (1 + \jj h\eta_{\sqrt{N_\mathrm{e}}-1,n})u_{\sqrt{N_\mathrm{e}}-1,n}^h.
\end{align}
\noindent {We note that we obtain the scattered field on $\Omega$ by directly truncating $\vu^{sc,h}$ from $\Omega_\mathrm{e}$ to $\Omega$ after solving \eqref{eq:HelmScatterRightIncident:discrete}.}


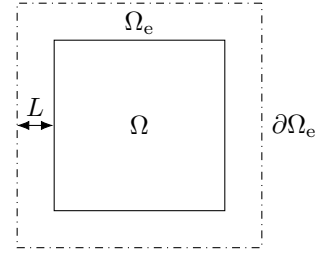
\begin{figure}
	\centering
	\begin{tikzpicture}[square/.style={regular polygon,regular polygon sides=4}]
	\node[square,draw,dashdotted,inner sep=1.15cm] (A) at (0,0){};
    \node[square,draw,solid,inner sep=0.675cm] (B) at (0,0){{$\Omega$}};
    \path [{Latex}-{Latex},draw] (A.west) -- node [midway,above] {{$L$}} (B.west);
    \node[below] at (A.north) {$\Omega_\mathrm{e}$};
    \node[right] at (A.east) {{$\partial\Omega_\mathrm{e}$}};
	\end{tikzpicture}
	\caption{A 2D domain with an absorbing boundary layer.}
	\label{fig:HelmholtzBoundaryCond}
\end{figure}
\subsection{Multigrid Methods and the Helmholtz Equation}\label{sec:MGSolverHelmholtz:sub:MGMethodsandHelmholtz}
Despite the apparent simplicity of MG methods, their direct application to the Helmholtz equation is not straightforward.
The reasons are two-fold: 1) the commonly used local smoothers (\eg{} pointwise smoothers) \cite{elman2001multigrid} in MG methods will diverge if applied to the Helmholtz equation; 2) the standard CGC will amplify certain components of the error instead of reducing them~\cite{ernst2012difficult}.
To understand these behaviors, we deploy the LFA tool to quantitatively estimate the performance of a two-grid cycle for a partial differential equation with constant coefficients \cite{wienands2004practical}.
To apply LFA, we momentarily assume that the object $k^2= k_0^2\eta^2(\vx) > 0$ is constant and that the boundary condition is periodic. Although such assumptions do not hold for the problem of interest here, we still gain relevant insights.

Let $\ve^h_\nu = \vu^h-\vv^h_\nu$ denote the error after calling $\nu$ times \Cref{alg:TG}.
Then, the error $\ve^h_{\nu+1}$ is specified by
\begin{eqnarray}
&\ve^h_{\nu+1}&=(\mS^h)^{v_2}\mC_h^{2h}(\mS^h)^{v_1}\ve^h_\nu\label{eq:itererror}\nonumber\\
\text{with }& \mS^h &= \mI_{N_\mathrm{e}} - \omega_\mathrm{S}\mD_{\mA^h_\mathrm{Hel}}^{-1}\mA^h_\mathrm{Hel}\label{eq:SmoothIteration}\\
\text{and }& \mC_h^{2h}&=\mI_{N_\mathrm{e}} - \mP(\mA^{2h})^{-1}\mR\mA^h_\mathrm{Hel},
\label{eq:TGIterationMatrix}
\end{eqnarray}
where $\mS^h\in \mathbb C^{{N_\mathrm{e}}\times {N_\mathrm{e}}}$ represents the iteration matrix of a local smoother at mesh-size~$h$, \eqref{eq:SmoothIteration} denotes the Jacobi method with damped factor $\omega_\mathrm{S}$, and
$\mC_h^{2h}\in \mathbb C^{{N_\mathrm{e}}\times {N_\mathrm{e}}}$~corresponds to the CGC.

Denote by $v^h(\vtheta,\vx)= \mathrm{e}^{\jj\theta_1 x_1/h} \mathrm{e}^{\jj\theta_2 x_2/h}$ a grid function and 
$\vv^h(\vtheta)\in\mathbb C^{{N_\mathrm{e}}}$ its discrete and vectorized counterpart
sampled on $\Omega_\mathrm{e}^h$. 
The parameters $\vtheta=(\theta_1,\theta_2)\in [-\pi,\pi)^2$ characterize the frequency of the grid function.
Under the assumption of a constant sample and periodic boundary conditions,
$\mA^h_\mathrm{Hel}$ and $\mS^h$ encode a circular convolution, which means that their eigenfunctions are the discrete grid functions.
The eigenvalues (modes) of $\mA^h_\mathrm{Hel}$ and $\mS^h$ are then $ {a}^h(\vtheta)=\frac{4-2(\cos\theta_1+\cos\theta_2)-k^2h^2}{h^2}$ and $ s^h(\vtheta)=\left(1-\omega_\mathrm{S}+\frac{2\omega_\mathrm{S}}{4-k^2h^2}(\cos\theta_1+\cos\theta_2)\right)$, respectively.
In the following, we briefly discuss two main challenges that hinder a direct application of MG methods to the Helmholtz equation and refer the readers to several works which fully present these issues~\cite{elman2001multigrid,ernst2012difficult}.

\paragraph{Divergence of the Local Smoothers}
For $(4-k^2h^2)>0$, we have that
\begin{equation}
\max_{\vtheta}| s_h(\vtheta)| = 1-\omega_\mathrm{S}+\frac{4\omega_\mathrm{S}}{4-k^2h^2},\end{equation}
which is always larger than $1$ when $kh\neq 0$. The Jacobi method will therefore be divergent when applied to the Helmholtz equation.
A similar phenomenon is also observed for other point-wise smoothers, such as Gauss-Seidel and its variants.
\paragraph{Amplification of the Error by the CGC}
The purpose of the CGC is to reduce the low-frequency (smooth) error, but it was observed that this step can amplify certain modes instead~\cite{elman2001multigrid}.
Let us assume that the current error on the fine problem after a pre-relaxation step consists in a smooth component~$\vv^h(\vtheta)$.
Then, the error after CGC reads \cite{yavneh1998coarse}
\begin{align}
    \ve^h &=\left(\mI^h-\mP(\mA^{2h}_\mathrm{Hel})^{-1}\mR\mA^h_\mathrm{Hel}\right)\vv^h(\vtheta)\nonumber\\
    &=\vv^h(\vtheta) - {a}^h(\vtheta)\mP(\mA^{2h}_\mathrm{Hel})^{-1}\mR\vv^h(\vtheta)\nonumber\\
    &\approx\left(1 - \frac{{a}^h(\vtheta)}{{a}^{2h}(2\vtheta)}\right)\vv^h(\vtheta),
\end{align}
where ${a}^{2h}(2\vtheta)$ is the eigenvalue of $\mA^{2h}_\mathrm{Hel}$ corresponding to $2\vtheta$, assuming that~$\mP\mR\vv^h(\vtheta) = \vv^h(\vtheta)$ for $\vtheta\in(-\frac{\pi}{2},\frac{\pi}{2}]^2$.
Evidently, the CGC will effectively reduce the error if the ratio $\frac{{a}^h(\vtheta)}{{a}^{2h}(2\vtheta)}$ is close to but less than $1$ and $\ve^h=\bm0$ if $\frac{{a}^h(\vtheta)}{{a}^{2h}(2\vtheta)}=1$. However, for the Helmholtz equation, prior works observed that~$\frac{{a}^h(\vtheta)}{{a}^{2h}(2\vtheta)}$ can be negative for certain components, especially on coarser grid~\cite{elman2001multigrid,ernst2012difficult}.
In those cases, the CGC amplifies the error since $\left(1 - \frac{{a}^h(\vtheta)}{{a}^{2h}(2\vtheta)}\right)>1$.
{Hence, this phenomenon will happen for more components if many levels are used}.

To overcome the divergence of the local smoothers, Brandt \emph{et al.} \cite{brandt1997wave} suggested to use the Kaczmarz method in the relaxation step.
This local smoother is convergent but converges slowly because it works on the normal equation.
Similarly, Elman \emph{et al.}~\cite{elman2001multigrid} used a convergent Krylov-based method as the local smoother, but their method nevertheless needs to store some previous iterates, thus increasing the memory requirement.
To solve the problem of the CGC, Stolk \emph{et al.}~\cite{stolk2014multigrid} proposed an optimized scheme to discretize the Helmholtz equation at the coarser levels.
Their method decreased the number of modes that lead to divergence, which enables the use of more levels.
However, the optimized schemes are the solutions of constrained minimization problems which must be resolved whenever the scattering potential changes.
Recent works showed that MG methods converge more easily if the lefthand side of~\eqref{eq:HelmScatterRightIncident} is $\left(-\nabla^2 - \kappa k_0^2\eta^2(\vx)\right)$ with $\kappa\in\mathbb C$  instead~\cite{erlangga2006novel}.
Let $\mK_\kappa$ denote the discretization of $\left(-\nabla^2 - \kappa k_0^2\eta^2(\vx)\right)$ on $\Omega^h_\mathrm{e}$.
The solution of~\eqref{eq:HelmScatterRightIncident:discrete} is then computed by using $\mK_\kappa$ as a preconditioner.
This technique, called shifted-Laplacian preconditioner, can help Krylov-based methods to converge faster~\cite{erlangga2006novel}.
We note that we did not find benefit in using the shifted-Laplacian preconditioner for the inverse-scattering problems presented in this paper.

\subsection{Proposed Multigrid-Based Solver}
\label{sec:MGSolverHelmholtz:sub:ProposedMGMethod}
In the spirit of~\cite{elman2001multigrid, erlangga2006novel}, we use Bi-CGSTAB with a preconditioner~\mbox{$\mK_\mathrm{MG} \approx \mA^h_\mathrm{Hel}$}
to solve~\eqref{eq:HelmScatterRightIncident:discrete}~(\Cref{alg:SolverHel}).
The efficiency of our method stems from the way we apply~$\mK_\mathrm{MG}^{-1}$: We deploy a standard MG method~(\Cref{alg:MG})~(see Steps~$9$ and $14$ in \Cref{alg:SolverHel}).
For the relaxation, we still use \eqref{eq:dampedJ} but with few iterations~(\ie{} $\nu_{1,2} \leq 2$) to mitigate a possible divergence of the local smoother. By doing so, Bi-CGSTAB would correct any deviation of the MG method~\cite{elman2001multigrid,erlangga2006novel}. Furthermore, we alleviate the issue of the CGC previously mentioned by using few levels.
In~\cite{calandra2013improved}, the best performance is achieved with two levels only, which corroborates what we observed in our experiments.
To mitigate the so-called pollution effect, a rule of thumb is to use at least $10$ points per wavelength for the coarsest level but slightly fewer than $10$ points per wavelength were sufficient in most of our experiments \cite{erlangga2006novel}.
In this work, we solve the coarsest level exactly, but one can also use iterative methods \cite{calandra2013improved,treister2019multigrid}.

For the restriction $\mR$, we use the full-weighting operator.
Specifically, the value of $r^{2h}_{m,n}=(\mR \vr^{h})_{m,n}$ is given by
\begin{align}
r^{2h}_{m,n} = &\frac{1}{16}\left(4r^h_{2m,2n} + 2\left(r^h_{2m-1,2n}+r^h_{2m+1,2n}\right.\right.\nonumber\\
&+ \left.r^h_{2m,2n-1}+r^h_{2m,2n+1}\right)
+ \left(r^h_{2m-1,2n-1}\right.\nonumber\\
&+r^h_{2m-1,2n+1} +\left.\left. r^h_{2m+1,2n-1}+r^h_{2m+1,2n+1}\right)\right),
\end{align}
where $m,n=0,1,\ldots,\frac{\sqrt{N_\mathrm{e}}-1}{2}$ denote the indices on the coarse problem. Note that the value of the points at the boundary is set to $0$. The prolongation $\mP$ is the adjoint of $\mR$ such that $e^{h}_{m,n} = (\mP\ve^{2h})_{m,n}$ is given by
\begin{align}
&e^h_{m,n} =\nonumber\\
&\left\{
\begin{array}{ll}
e^{2h}_{m/2,n/2},& m,n ~\text{even}\\
\frac{1}{2}\left(e^{2h}_{(m-1)/2,n/2}+e^{2h}_{(m+1)/2,n/2}\right),& m~\text{odd},n~\text{even}\\
\frac{1}{2}\left(e^{2h}_{m/2,(n-1)/2}+e^{2h}_{m/2,(n+1)/2}\right),& m~\text{even},n~\text{odd}\\
\frac{1}{4}\left(e^{2h}_{(m-1)/2,(n-1)/2}+e^{2h}_{(m-1)/2,(n+1)/2}\right.\\
+\left.e^{2h}_{(m+1)/2,(n-1)/2}+e^{2h}_{(m+1)/2,(n+1)/2}\right),& m,n~\text{odd},
\end{array}\right.
\end{align}
where $m,n=0,1,\ldots,\sqrt{N_\mathrm{e}}-1$ denote the indices on the fine problem. 
For the coarser problems, we directly re-discretize the Laplacian operator with double mesh-size and use the full-weighted transfer to restrict $k_0^2\eta^2(\vx)$. The value of the points near the boundary of $k_0^2\eta^2(\vx)$ is set to $k_0^2\eta_\mathrm{b}^2$.

\begin{algorithm}[t]
	\caption{Bi-CGSTAB with \Cref{alg:MG} as a preconditioner for solving \eqref{eq:HelmScatterRightIncident:discrete}}
	\label{alg:SolverHel}
	\begin{algorithmic}[1]
	    \REQUIRE {Set $\mK \leftarrow  -\nabla^2-k_0^2\eta^2(\vx)$ and $CycleType,~h,~N_\mathrm{level}$ for \Cref{alg:MG} and choose tolerance  $\varepsilon$.}
   \STATE {$\vr_0 = \vb^h-\mA^h_\mathrm{Hel} \vu_0^{\mathrm{sc},h}$.}
   \STATE {$\hat {\vr}_0= \vr_0$.}
   \STATE {$\rho_0= 1,~\alpha=1,~\sigma_0= 1$.}
   \STATE {$\vv_0= \bm0,~\vp_0= \bm 0$.}
   \STATE {$\beta=0$, $\vy=\mathbf{0}$, $\vh=\mathbf{0}$,
   $\vs=\mathbf{0}$, $\vz=\mathbf{0}$, $\vt=\mathbf{0}$.}
      \FOR {$Iter=1,2,\cdots$} 
      \STATE {$\rho_{Iter}= <\hat{\vr}_0,\vr_{Iter-1}>$}  
      \STATE {$\beta\leftarrow \left(\rho_{Iter}/\rho_{Iter-1}\right)\left(\alpha /\sigma_{Iter-1}\right)$.}
      \STATE {$\vp_{Iter}=\vr_{Iter-1}+\beta(\vp_{Iter-1}-\sigma_{Iter-1}\vv_{Iter-1})$.}
      \STATE  
      {\mybox{$\vy\leftarrow MGCycle(\mK,\vp_{Iter},\bm 0,CycleType,h,N_\mathrm{level}).$}} 
      \STATE {$\vv_{Iter}= \mA^h_\mathrm{Hel}\vy$.}
      \STATE {$\alpha\leftarrow \rho_{Iter}/\left<\hat{\vr}_0,\vv_{Iter}\right>$.}
      \STATE {$\vh \leftarrow  \vu_{Iter-1}^{\mathrm{sc},h}+\alpha \vy$.}
      \STATE {$\vs\leftarrow \vr_{Iter-1}-\alpha \vv_{Iter}$.}
      \STATE
      {\mybox{$\vz\leftarrow MGCycle(\mK,\vs,\bm 0,CycleType,h,N_\mathrm{level}).$}\label{alg:SolverHel:PredII}}
      \STATE {$\vt\leftarrow \mA^h_\mathrm{Hel}\vz$.}
      \STATE {$\sigma_{Iter}= \left<\vt,\vs\right>/\left<\vt,\vt\right>$.}
      \STATE {$\vu^{\mathrm{sc},h}_{Iter}= \vh+\sigma_{Iter}\vz$.}
      \STATE {$\vr_{Iter}=  \vs-\sigma_{Iter}\vt$.}
       \IF {$\|\vr_{Iter}\|_2\leq \varepsilon$.}
       \STATE Return $\vu^{\mathrm{sc},h}_{Iter}$.
       \ENDIF
        \ENDFOR
	\end{algorithmic}
\end{algorithm}

\section{Problem Formulation and Optimization}
\label{sec:RecProblemFormulationandOptimization}
We are now equipped with the Helmholtz-based forward model
\begin{align}
\mH_\mathrm{MGH}:\vf &\mapsto \tilde{\mG}\diag(\vf)(\vu^\mathrm{sc}_{\mathrm{MGH}}(\vf) + \vu^\mathrm{in}), \label{eq:DiscreteForwardHelm}
\end{align}
where $\vu^\mathrm{sc}_{\mathrm{MGH}}(\vf)$ is computed with \Cref{alg:SolverHel}.
Then, we formulate inverse-scattering as the solution~$\vf^\ast$ to a composite problem with nonnegativity constraint
\begin{equation}
\vf^{\,*}=\arg\min_{\vf\in \mathbb{R}^N_{\geq 0}} \sum_{q=1}^Q\mathfrak D_q(\mH_\mathrm{MGH}^q(\vf),\vy^\mathrm{sc}_q) + \tau \mathfrak R(\vf),\label{eq:MAPInverse}
\end{equation}
where the data-fidelity term $\mathfrak D_q:\mathbb C^M\times \mathbb C^M \rightarrow \mathbb R$ enforces the consistency with the measurements~$\{\vy^\mathrm{sc}_q\in\mathbb{C}^M\}_{q=1}^Q$,
$\mathfrak R:\mathbb R^N\rightarrow \mathbb R$ regularizes the solution, and $\tau>0$ is a tradeoff parameter to balance these two terms. Note that the forward model~$\mH_\mathrm{MGH}^q(\vf)$ uses the incident field $\vu^\mathrm{in}_q$. 
In this work, we set the data-fidelity term as the quadratic error, for $q=1,\ldots,Q$,
\begin{equation}
\mathfrak D_q(\mH_\mathrm{MGH}^q(\vf),\vy^\mathrm{sc}_q) =\frac{1}{2}\|\mH_\mathrm{MGH}^q(\vf)-\vy_q^\mathrm{sc}\|_2^2.\label{eq:DataFidelityExplicit}
\end{equation}
For the regularization term, we choose the isotropic total variation (TV)~\cite{rudin1992nonlinear} but one can adopt other regularizations such as the Hessian Schatten-norm \cite{lefkimmiatis2013hessian}, plug-and-play prior \cite{kamilov2017plug,hong2020solving}, or tailored regularization \cite{pham2020adaptive}.

The accelerated forward-backward splitting (FBS)~\cite{beck2009fast,nesterov2013gradient} is adopted here to solve \eqref{eq:MAPInverse}. The detailed description of FBS is summarized in \Cref{alg:FISTA}, of which we provide now some details.
\begin{itemize}
\item If $Q$ is large enough, then one may use the stochastic version as shown at Line \ref{alg:FISTA:stochasticSelect} so that only a (random) subset of the measurements is chosen to estimate the gradient of the data-fidelity term at each iteration to reduce the computational burden \cite{soubies2017efficient,pham2020three}.
\item At Line~\ref{alg:FISTA:proximal}, $\text{prox}_{\gamma_\nu\tau}(\vw_\nu)$ denotes the proximal operator evaluated as
\begin{align}
\text{prox}_{\gamma_\nu\tau}(\vw_\nu) = \arg\min_{\vw\in\mathbb{R}^N_{\geq 0}}\Bigg(\frac{1}{2}&\|\vw-\vw_\nu\|_2^2 \nonumber\\
+& \left(\gamma_\nu\tau\right)\mathfrak{R}(\vw)\Bigg).\label{eq:proximalOpt}
\end{align}
Since $\mathfrak{R}(\cdot)$ is TV in our case and $\vw$ is nonnegative, there is no closed-form solution for \eqref{eq:proximalOpt}.
We therefore consider \eqref{eq:proximalOpt} with the fast gradient projection on its dual formulation to address the non-smoothness of TV \cite{beck2009fastTV}.
\item We optimize a non-convex problem because the forward model is nonlinear. 
To the best of our knowledge, there exists no theoretical proof of the global convergence of the accelerated FBS for non-convex problems.
However, we observed that \Cref{alg:FISTA} behaves well for our problem.
The stepsize is empirically set to a fixed value.
\end{itemize}
\begin{algorithm}[t]
 \caption{Accelerated FBS to solve \eqref{eq:MAPInverse}  \cite{beck2009fast,nesterov2013gradient}}            
 \label{alg:FISTA}                 
 \begin{algorithmic}[1]
 \REQUIRE ~\\
 $\vf^{\,0}\in\mathbb{R}^N_{\geq 0},$ stepsize $\gamma_\nu>0$, and $\nu$ is the iteration index. 
 \lastcon $\vf^{\,*}$.
 \STATE $\vv^{1} = \vf^{\,0}.$
 \STATE $\alpha_1 = 1.$
 \STATE $\nu \leftarrow 1.$
 \WHILE{not converged}
 \STATE {Select a subset $\tilde{Q}\subseteq [1\ldots Q].$\label{alg:FISTA:stochasticSelect}}
 \STATE {$\vd^{\nu}=  \sum_{q\in \tilde{Q}} \nabla_{\vf}\mathfrak D_q(\bar{\vf}^{\nu}).$ \label{alg:FISTA:gradient}}
 \STATE {$\vf^{\,\nu}= \text{prox}_{\gamma_\nu\tau}(\bar{\vf}^{\,\nu} - \gamma_\nu\vd^{\nu}).$\label{alg:FISTA:proximal}}
 \STATE $\alpha_{\nu+1}= \frac{1+\sqrt{1+4\alpha_\nu^2}}{2}.$
 \STATE $\bar{\vf}^{\,\nu}= \vf^{\,\nu}+\frac{\alpha_\nu-1}{\alpha_{\nu+1}}\left(\vf^{\,\nu}-\vf^{\,\nu-1}\right).$
 \STATE $\nu\leftarrow \nu+1.$
 \ENDWHILE
 \STATE $\vf^{\,*}= \vf^{\,\nu}.$
 \end{algorithmic}
 \end{algorithm}
The evaluation of the gradient at Line \ref{alg:FISTA:gradient} requires the Jacobian matrix of $\mH_\mathrm{MGH}^q$ which is specified in \Cref{prop:HelmholtzJacobian}. With this formulation, the evaluation of the gradient of the data-fidelity term for the Helmholtz model mainly costs one inversion of the matrix~${\mA^h_\mathrm{Hel}}$, which is again efficiently performed with \Cref{alg:SolverHel}.
\begin{prop}\label{prop:HelmholtzJacobian}
	The Jacobian matrix of~$\mH_\mathrm{MGH}^q$
	\begin{equation}
	\mJ_{\mH_\mathrm{MGH}^q}(\vf)=\left(\mI+\diag(\vf)({\mA^h_\mathrm{Hel}})^{-1}\right)\diag\left(\vu_q(\vf)\right).
	\end{equation}
\end{prop}
\begin{proof}
Similar to the derivation of the Jacobi matrix of the LiS model in \cite{soubies2017efficient}, the G\^{a}teaux derivative in the direction $\vv\in\mathbb{R}^{N_\mathrm{e}}$ is
\begin{equation}
\begin{array}{rcl}
\mathrm{d}\mH_\mathrm{MGH}^q(\vf;\vv)&=&\lim\limits_{\epsilon\rightarrow 0}\frac{\diag(\vf+\epsilon\vv)\vu_q(\vf+\epsilon\vv)-\diag(\vf)\vu_q(\vf)}{\epsilon}\\
&=&\diag(\vu_q(\vf))\vv \\
&\quad&\quad+\lim\limits_{\epsilon\rightarrow 0}\diag(\vf)\frac{\vu_q(\vf+\epsilon\vv)-\vu_q(\vf)}{\epsilon}.
\end{array}
\end{equation}
Then, for $\vf\rightarrow \vf+\epsilon\vv$, $k_0^2\eta^2(\vx)$ in \eqref{eq:HelmScatterRightIncident} becomes $k_0^2\eta^2(\vx)+\epsilon\diag(\vv)$, which yields
\begin{align}
\vu_q(\vf)&=\vu^\mathrm{in}_q+({\mA^h_\mathrm{Hel}})^{-1}\diag(\vf)\vu_q^\mathrm{in},\nonumber\\
\vu_q(\vf+\epsilon\vv)&=\vu_q^\mathrm{in}+\left(\mA^h_\mathrm{Hel}-\epsilon\diag(\vv)\right)^{-1}
\diag(\vf+\epsilon\vv)\vu_q^\mathrm{in}.
\end{align}
Then, we have that
 \begin{equation}
 \label{eq:proofdiff}
 \vu_q(\vf+\epsilon\vv)-\vu_q(\vf)=\left(\mA^h_\mathrm{Hel}-\epsilon\diag(\vv)\right)^{-1}\diag(\epsilon\vv)\vu_q(\vf).
\end{equation}
Substituting \eqref{eq:proofdiff} into $\mathrm{d}\mH^q_\mathrm{MGH}(\vf;\vv)$ and taking the limit, we get the desired result
\begin{equation}
\begin{array}{rl}
\mathrm{d}\mH^q_\mathrm{MGH}(\vf;\vv)=&\diag(\vu_q(\vf))\vv+\diag(\vf)({\mA^h_\mathrm{Hel}})^{-1}\\
&\times\diag(\vu_q(\vf))\vv\\
 = & \left(\mI+\diag(\vf)({\mA^h_\mathrm{Hel}})^{-1}\right)\diag(\vu_q(\vf)) \vv.
\end{array}
\end{equation}
\end{proof}

\section{Numerical Experiments}\label{sec:Exps}
In the first set of experiments, we compare the total fields obtained by the LiS and Helmholtz models.
We choose samples for which analytical solutions exist.
In the second set of experiments, we compare the performance of the LiS and Helmholtz models on an inverse-scattering problem with simulated and real data. Note that Bi-CGSTAB is used to solve~\eqref{eq:invertLSm} (MATLAB built-in function \emph{bicgstab}).
The whole implementation is based on GlobalBioIm~\cite{soubies2019pocket} and was performed on a laptop with Intel Core i$9$ $2.3$GHz.
The algorithm for \eqref{eq:invertLSm} and \eqref{eq:HelmScatterRightIncident:discrete} is said to have converged when the relative error reaches~$10^{-6}$.

Despite that one can take advantage of parallelization~\cite[Chapter $6$]{trottenberg2000multigrid} or GPU acceleration~\cite{knibbe2011gpu} for MG methods,
we have implemented our MG method in MATLAB without parallelization or GPU acceleration.
In return, when several threads are available,
the LiS method takes advantage of the parallelized implementation of the FFT in MATLAB.
Thus, to provide a fair comparison, we run both MG and LiS methods with only one CPU thread.

In our experiments, we run one V-cycle to apply~$\mK_{MG}^{-1}$ and perform one pre- and post-relaxation ($\nu_1=\nu_2=1$). We choose the damped Jacobi relaxation with $\omega_\mathrm{S}=0.8$ as the local smoother. Moreover, additional $\frac{\sqrt{N}}{8}$ points are added at each side as the ABL (\ie{} ${N_\mathrm{e}}=\frac{25N}{16}$) with $\beta=0.15$ for the first set of experiments.
For the inverse-scattering problems, we use $\frac{\sqrt{N}}{16}$ points as the ABL and set $\beta=0$.

\subsection{Robustness and Efficiency}
\label{sec:sub:RobuEffMG}
In this part, we consider a disk with RI $\eta_\mathrm{disk}$ immersed in air~($\eta_\mathrm{b}=1$, \Cref{fig:beadsetting}).
For such objects, there exists an analytic expression of the total field~\cite{devaney2012mathematical},
which allows us to study the accuracy of the total field obtained by the LiS and Helmholtz methods.
\begin{figure}
    \centering
    \begin{tikzpicture}
    \newcommand{\insep}{2}
    \newcommand{\foldoi}{figs/Exp1_1.00_2.20_rad_1.25}
    \newcommand{\mywidth}{0.25*\textwidth}
    \pgfmathsetmacro{\winsz}{5*256/406}
    \begin{groupplot}[
     group style = {group size = 2 by 2,
     vertical sep=\insep, horizontal sep=\insep},
     xmin = 0,xmax = \winsz, ymin = 0, ymax = \winsz,
    enlargelimits=false,
    axis equal image,
    scale only axis,
    width = {\mywidth},
    hide axis,
    title style = {anchor=base}  
    ]
    \nextgroupplot[
        xmin=0, xmax=\winsz, ymin=0, ymax=\winsz,
        axis equal image,
        hide axis
        ]
    \addplot graphics[xmin=0,xmax=\winsz,ymin=0,ymax=\winsz,includegraphics={trim=0.72cm 0.53cm  1.43cm 0.11cm, clip}] {{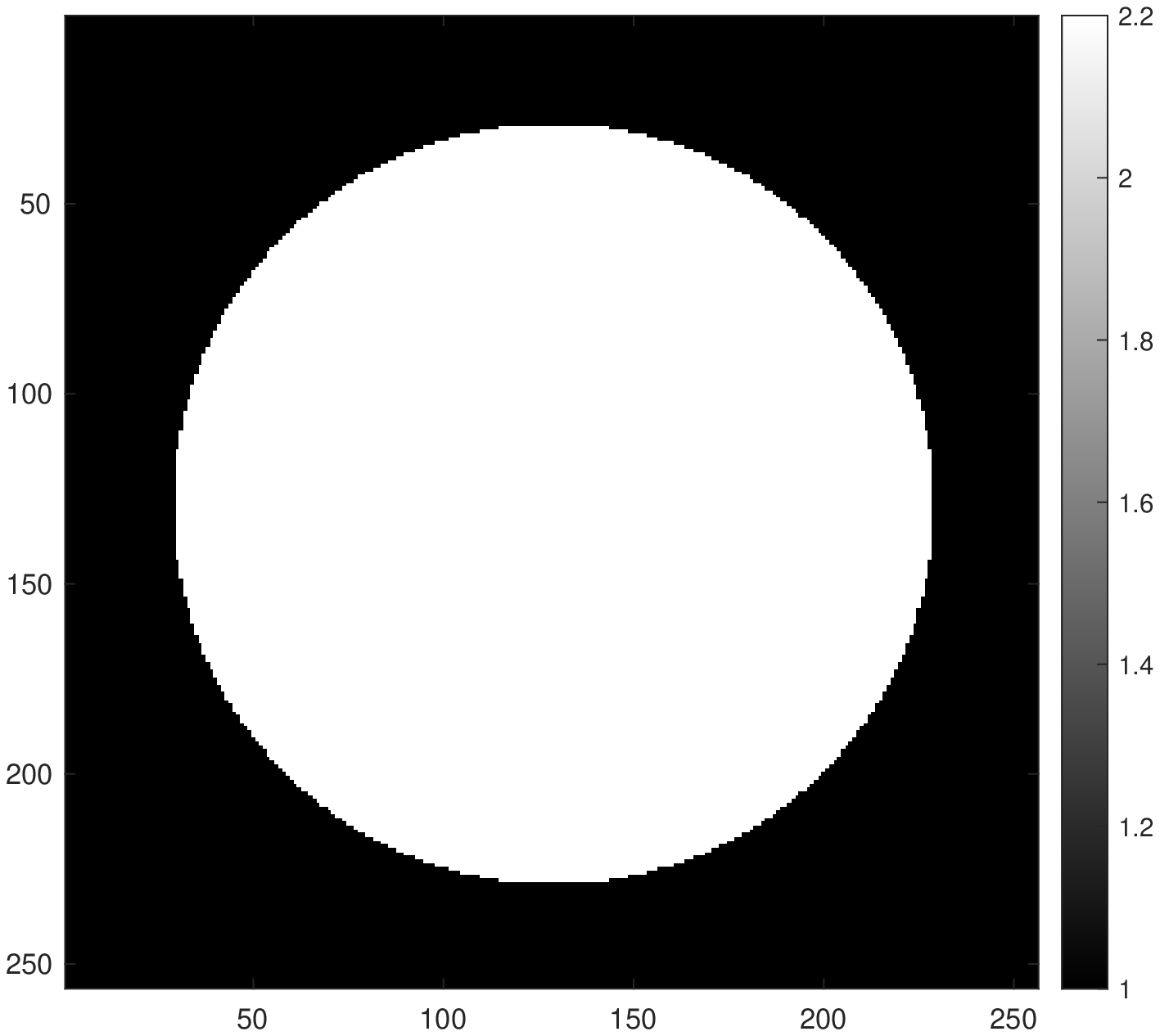}};
    \node[anchor=north west,white,text width=6cm] at (axis cs:0.1,1.01*\winsz) {\textbf{Disk}};
    \node[white,anchor=north east] at (axis cs:\winsz,1.01*\winsz) {${\eta_\mathrm{b}=1}$};
    \node at (axis cs:\winsz/2,\winsz/2) {${\eta_\mathrm{disk}=2.2}$};
    
    \nextgroupplot
    \addplot graphics[xmin=0,xmax=\winsz,ymin=0,ymax=\winsz,includegraphics={trim=0.72cm 0.53cm  1.43cm 0.11cm, clip}] {{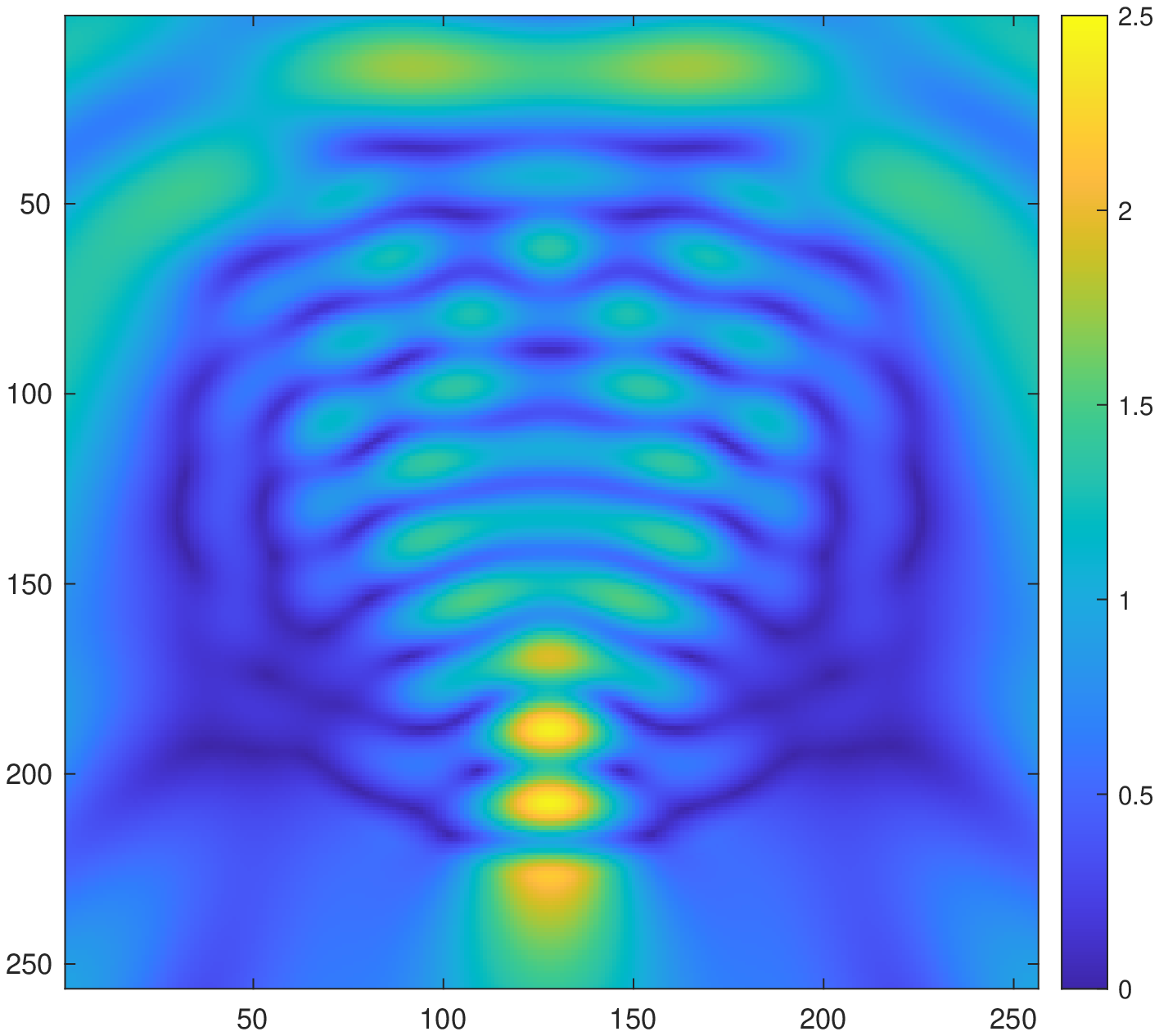}};
    \node[anchor=north west,white,text width=6cm] at (axis cs:0.1,1.01*\winsz) {\textbf{Ground Truth}};
    \fill [white] (axis cs:0.1,0.2) rectangle (axis cs:1.1,0.25);
    \draw [white] (axis cs:0.6,0.25) node[anchor=south] {$\lambda$};
    \draw [dashed,white,thick] (axis cs:\winsz/2,\winsz/2) circle[radius=1.23152709359606];

    \nextgroupplot[enlargelimits=false]
    \addplot graphics[xmin=0,xmax=\winsz,ymin=0,ymax=\winsz,includegraphics={trim=0.72cm 0.53cm  1.43cm 0.11cm, clip}] {{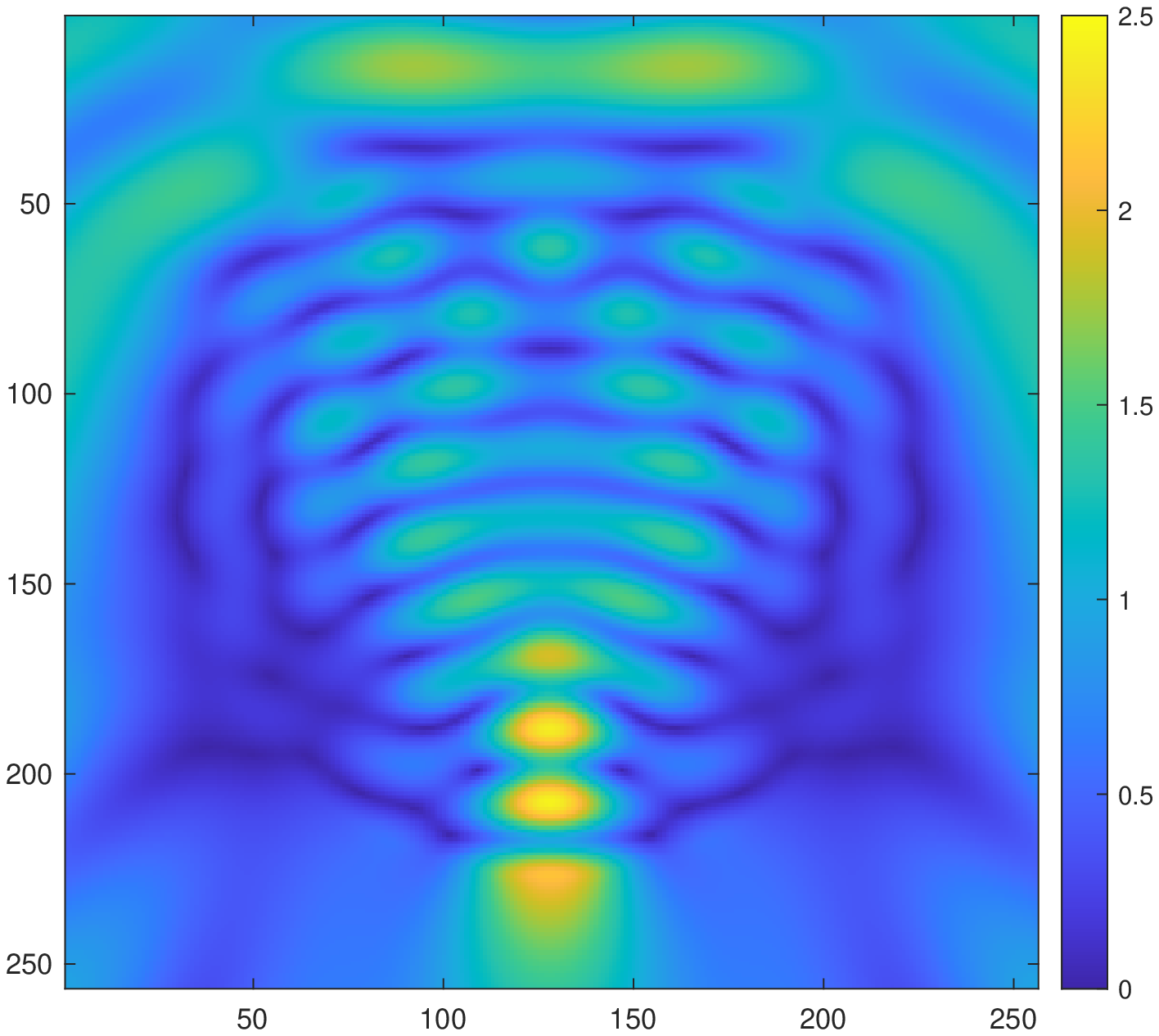}};
    \node[anchor=north west,white,text width=6cm] at (axis cs:0.1,1.01*\winsz) {\textbf{LiS}};
    \draw [dashed,white,thick] (axis cs:\winsz/2,\winsz/2) circle[radius=1.23152709359606];
    
    \nextgroupplot[enlargelimits=false]
    \addplot graphics[xmin=0,xmax=\winsz,ymin=0,ymax=\winsz,includegraphics={trim=0.72cm 0.53cm  1.43cm 0.11cm, clip}] {{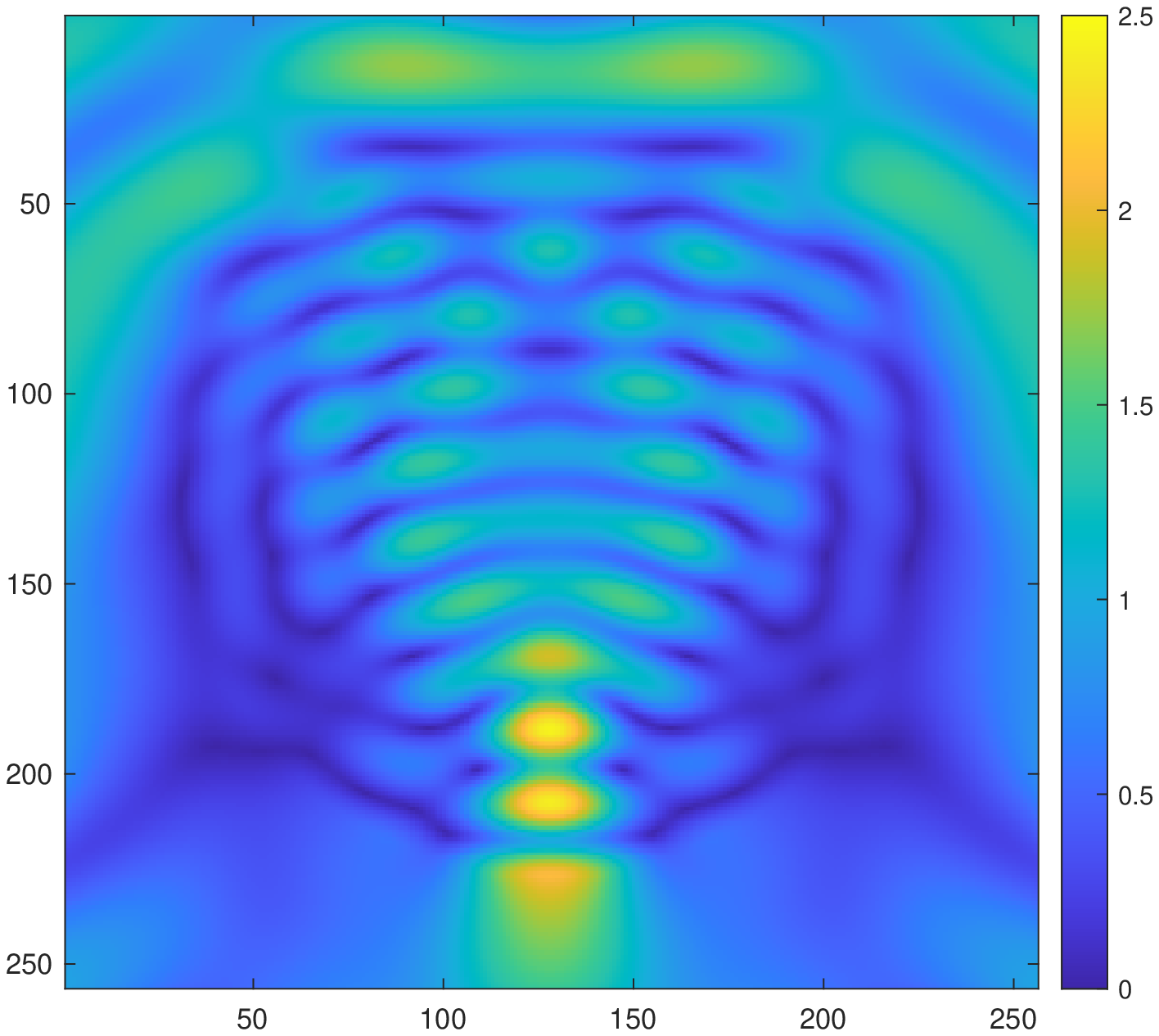}};
    \node[anchor=north west,white,text width=6cm] at (axis cs:0.1,1.01*\winsz) {\textbf{MGH}};
    \draw [dashed,white,thick] (axis cs:\winsz/2,\winsz/2) circle[radius=1.23152709359606];
    \end{groupplot}
    \end{tikzpicture}
    \caption{Total fields of a disk with radius $= 1.25\lambda$ and a RI of $\eta_\mathrm{disk}=2.2$.
    The disk is immersed in air~($\eta_\mathrm{b}=1$) and illuminated by a plane wave of wavelength~$\lambda=10$cm.
    The displayed fields are in a square area of length~$= 3.2\lambda$
    and are obtained through an analytical solution~\cite{devaney2012mathematical}, the LiS, and the Helmholtz methods with $N=256^2,~h=0.125$cm. The corresponding relative error of the LiS and Helmholtz models are $8.1\times 10^{-3}$ and $7.5\times10^{-3}$, respectively.}
    \label{fig:beadsetting}
\end{figure}

The disk is illuminated from the top by a plane wave of wavelength~$\lambda = 10$ cm. Our region of interest is a square area of length~$3.2\lambda$~(\Cref{fig:beadsetting}). A total of $N=256^2$ samples are used to discretize the domain~(\ie{} $h=0.125$ cm).
Denote by
$$
\epsilon = \frac{\|\vu-\vu_\mathrm{GT}\|^2}{\|\vu_\mathrm{GT}\|^2}
$$
the relative error where $\vu$ is the estimated total field and $\vu_\mathrm{GT}$ is the ground truth. From \Cref{fig:beadsetting}, one observes that both models yield an accurate total field with low relative error ($8.1\times 10^{-3}$ and $7.5\times 10^{-3}$ for the LiS and Helmholtz models, respectively).

To study the efficiency of the Helmholtz model with the proposed MG method,
we perform a series of experiments similar to the previous one, but with diverse sets of contrasts ($\max(|\vf|)/k_0^2\eta_\mathrm{b}^2$) and radii.
We adopt the same square domain, wavelength, RI of the background, and source position as were shown in~\Cref{fig:beadsetting}. Three levels are used for the MG method.

The number of iterations and the computational time to converge is provided in \Cref{fig:showRoubEffMie256}.
We see that the LiS model takes more time to converge when the contrast or the radius of the sample increases, which corresponds to the most challenging cases.
In comparison, the Helmholtz model constantly performs well,
which suggests that the proposed method is robust.

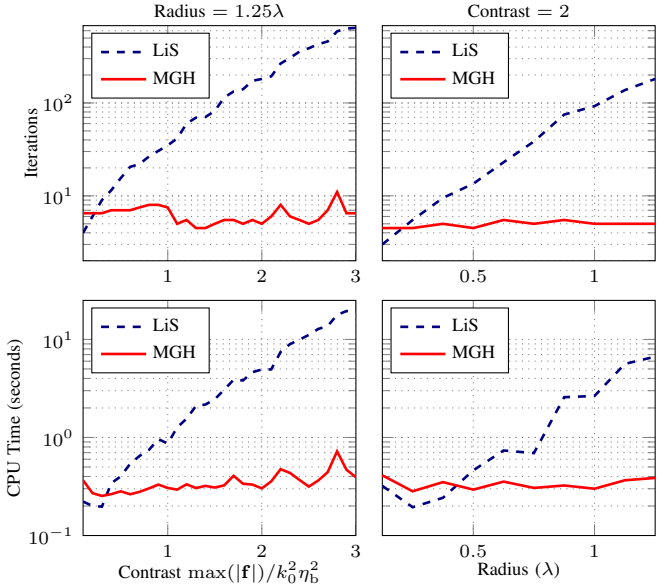
\begin{figure}
	\centering
	\newcommand{\yminIter}{2}
	\newcommand{\ymaxIter}{680}
	\newcommand{\yminTime}{0.1}
	\newcommand{\ymaxTime}{25}
	\begin{tikzpicture}
	\pgfplotsset{every axis legend/.append style={legend pos=north west,
anchor=north west,font=\scriptsize, legend cell align={left}}}
\pgfplotsset{grid style={dotted, gray}}
	
	\begin{groupplot}[enlargelimits=false,group style={group size=2 by 2,horizontal sep = 1em,vertical sep=1.5em},scale only axis,
	width=0.2*\textwidth,
	every axis/.append style={ymode=log,font=\scriptsize,title style={anchor=base,yshift=-1mm}, x label style={yshift = 0.5em}, y label style={yshift = -.5em},grid = both}]
	
	\nextgroupplot[xlabel = {},ylabel=Iterations,title={Radius $=1.25\lambda$},ymin=\yminIter,ymax=\ymaxIter]
	\addplot[dashed,darkblue,line width=1pt] table [search path={figs/Exp1_Contrast},x=ContrastLevel, y=LiSNumIteration, col sep=comma] {ResultsContrast256.csv};
	\addplot[solid,red,line width=1pt] table [search path={figs/Exp1_Contrast},x=ContrastLevel, y=HelmholtzIteration, col sep=comma] {ResultsContrast256.csv};
	\legend{LiS, MGH}
	
	\nextgroupplot[xlabel = {},ylabel={},title={Contrast $=2$},ymin=\yminIter,ymax=\ymaxIter,yticklabels={,,}]
	\addplot[dashed,darkblue,line width=1pt] table [search path={figs/Exp1_Radius},x=RadiusLevel, y=LiSNumIteration, col sep=comma] {ResultsRadius256.csv};
	\addplot[solid,red,line width=1pt] table [search path={figs/Exp1_Radius},x=RadiusLevel, y=HelmholtzIteration, col sep=comma] {ResultsRadius256.csv};
	\legend{LiS, MGH}
	
    \nextgroupplot[xlabel = Contrast~$\max(|\vf|)/k_0^2\eta_\mathrm{b}^2$, ylabel=CPU Time (seconds),title={},legend style={font=\scriptsize},ymin=\yminTime,ymax=\ymaxTime]
	\addplot[dashed,darkblue,line width=1pt] table [search path={figs/Exp1_Contrast},x=ContrastLevel, y=LiSCPUTime1-thread, col sep=comma] {ResultsContrast256.csv};
	\addplot[solid,red,line width=1pt] table [search path={figs/Exp1_Contrast},x=ContrastLevel, y=HelmholtzCPUTime1-thread, col sep=comma] {ResultsContrast256.csv};
	\legend{LiS,MGH}
	
    \nextgroupplot[xlabel = Radius~($\lambda$), ylabel={},title={},legend style={font=\scriptsize},ymin=\yminTime,ymax=\ymaxTime,yticklabels={,,}]
	\addplot[dashed,darkblue,line width=1pt] table [search path={figs/Exp1_Radius},x=RadiusLevel, y=LiSCPUTime1-thread, col sep=comma] {ResultsRadius256.csv};
	\addplot[solid,red,line width=1pt] table [search path={figs/Exp1_Radius},x=RadiusLevel, y=HelmholtzCPUTime1-thread, col sep=comma] {ResultsRadius256.csv};
	\legend{LiS,MGH}
    \end{groupplot}
	\end{tikzpicture}
	
	\caption{Number of iterations (top line) and CPU time (bottom line) \emph{versus} contrast (left column) and radius (right column) for the Lippmann-Schwinger and Helmholtz models.
	The domain is discretized with $N=256^2$ points and the mesh-size~$h=0.125$cm.}
	\label{fig:showRoubEffMie256}
\end{figure}

Next, we discretize the same domain with $N=1024^2$, which results in a large-scale problem.
From \Cref{fig:showRoubEffMie1024}, we see that Bi-CGSTAB for the LiS method requires more iterations to converge, which is similar to the phenomenon observed in \Cref{fig:showRoubEffMie256}.
Regarding the computational time, the LiS method can be $20$ times slower than for the case $N=256^2$ (\eg{} contrast or radius larger than $3$ or $1.2\lambda$, respectively).
On the contrary, the increase of the computational time of the Helmholtz method is moderate for $N=1024^2$ because we used more levels for the MG method.
Indeed, this feature improves the convergence speed at the price of a slightly increased computational cost, as discussed in \Cref{sec:MGSolverHelmholtz:sub:MGMethods}.

\begin{figure}
	\centering
	\newcommand{\yminIter}{2}
	\newcommand{\ymaxIter}{1.1e3}
	\newcommand{\yminTime}{1}
	\newcommand{\ymaxTime}{360}
	\begin{tikzpicture}
	\pgfplotsset{every axis legend/.append style={legend pos=north west,
anchor=north west,font=\scriptsize, legend cell align={left}}}
	\pgfplotsset{grid style={dotted, gray}}
	\begin{groupplot}[enlargelimits=false,group style={group size=2 by 2,horizontal sep = 1em,vertical sep=1.5em},
	width=0.2*\textwidth,
	scale only axis,
	every axis/.append style={ymode=log,font=\scriptsize,title style={anchor=base,yshift=-1mm}, x label style={yshift = 0.5em}, y label style={yshift = -.5em}, grid = both }]
	
	\nextgroupplot[xlabel = {},ylabel=Iterations,title={Radius $=1.25\lambda$},ymin=\yminIter,ymax=\ymaxIter]
	\addplot[dashed,darkblue,line width=1pt] table [search path={figs/Exp1_Contrast},x=ContrastLevel, y=LiSNumIteration, col sep=comma] {ResultsContrast1024.csv};
	\addplot[solid,red,line width=1pt] table [search path={figs/Exp1_Contrast},x=ContrastLevel, y=HelmholtzIteration, col sep=comma] {ResultsContrast1024.csv};
	\legend{LiS, MGH}
	
	\nextgroupplot[xlabel = {},ylabel={},title={Contrast $=2$},ymin=\yminIter,ymax=\ymaxIter,yticklabels={,,}]
	\addplot[dashed,darkblue,line width=1pt] table [search path={figs/Exp1_Radius},x=RadiusLevel, y=LiSNumIteration, col sep=comma] {ResultsRadius1024.csv};
	\addplot[solid,red,line width=1pt] table [search path={figs/Exp1_Radius},x=RadiusLevel, y=HelmholtzIteration, col sep=comma] {ResultsRadius1024.csv};
	\legend{LiS, MGH}
	
    \nextgroupplot[xlabel = Contrast $\max(|\vf|)/k_0^2\eta_\mathrm{b}^2$, ylabel=CPU Time (seconds),title={},legend style={font=\scriptsize},ymin=\yminTime,ymax=\ymaxTime]
	\addplot[dashed,darkblue,line width=1pt] table [search path={figs/Exp1_Contrast},x=ContrastLevel, y=LiSCPUTime1-thread, col sep=comma] {ResultsContrast1024.csv};
	\addplot[solid,red,line width=1pt] table [search path={figs/Exp1_Contrast},x=ContrastLevel, y=HelmholtzCPUTime1-thread, col sep=comma] {ResultsContrast1024.csv};
	\legend{LiS,MGH}
	
    \nextgroupplot[xlabel = Radius~($\lambda$), ylabel={},title={},legend style={font=\scriptsize},ymin=\yminTime,ymax=\ymaxTime,yticklabels={,,}]
	\addplot[dashed,darkblue,line width=1pt] table [search path={figs/Exp1_Radius},x=RadiusLevel, y=LiSCPUTime1-thread, col sep=comma] {ResultsRadius1024.csv};
	\addplot[solid,red,line width=1pt] table [search path={figs/Exp1_Radius},x=RadiusLevel, y=HelmholtzCPUTime1-thread, col sep=comma] {ResultsRadius1024.csv};
	\legend{LiS,MGH}
    \end{groupplot}
	\end{tikzpicture}
	
	
	\caption{Number of iterations (top line) and CPU time (bottom line) \emph{versus} contrast (left column) and radius (right column) for the Lippmann-Schwinger and Helmholtz models.
	The domain is discretized with $N=1024^2$ points and the mesh-size~$h=0.0312$cm.}
	\label{fig:showRoubEffMie1024}
\end{figure}
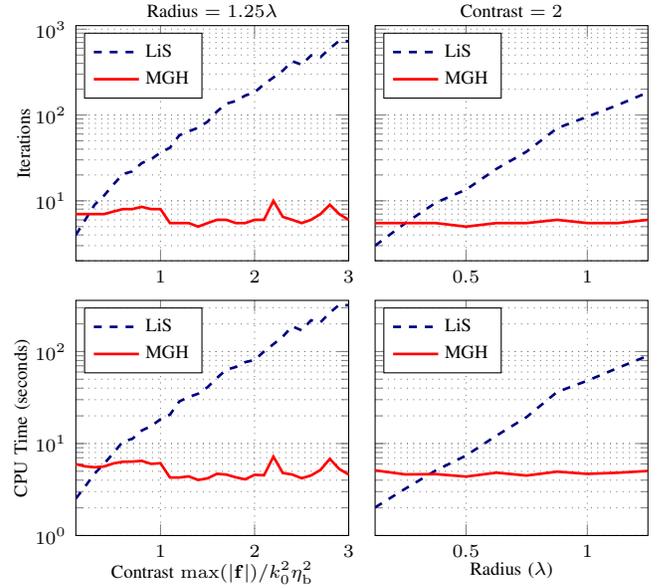

\subsection{Inverse Scattering with Simulated Data}
\label{sec:Exps:sub:SimulatedData}
In this section, we solve an inverse-scattering problem with simulated data.
We generated a synthetic image (\Cref{fig:SimulatedDataSample}) with contrast $0.355$ and size $4.5\lambda$, immersed in air ($\eta_\mathrm{b}=1$).
We illuminate the sample with plane waves of wavelength $\lambda=3$cm.
Simulations were conducted on a fine grid ($N=1024^2$) with square pixel of length~$4.4\times10^{-3}\lambda$ using the LiS and Helmholtz models.
We simulated $35$ illuminations that were uniformly distributed around the object and placed $360$ detectors around the object at a distance of $25$cm from the center, but recorded only the $120$ detectors that were the farthest from the illumination source. In total, we obtained $35\times 120$ measurements.

\begin{figure}
	\centering
	\includegraphics[scale=0.5]{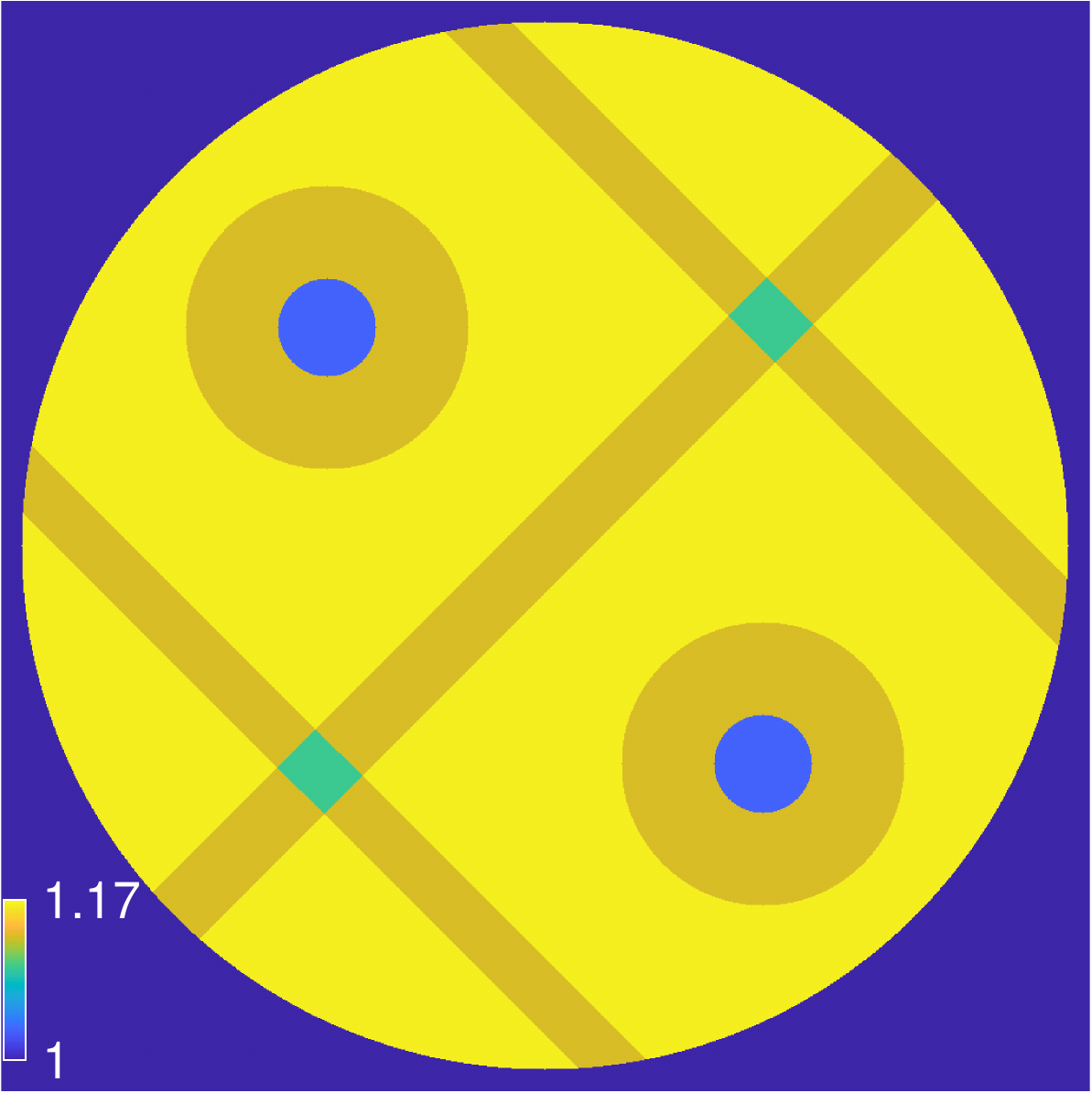}
	\caption{RI of the sample in the simulated experiment. The contrast is $35.5\%$.}
	\label{fig:SimulatedDataSample}
\end{figure}

For the reconstruction, we considered two different grids: $N=256^2$ with square pixel of length~$1.76\times 10^{-2}\lambda$ and $N=512^2$ with square pixel of length~$8.8\times 10^{-3}\lambda$.
For the reconstructed algorithm, $250$~iterations were performed. The stepsize $\gamma$ and regularization parameter $\tau$ are summarized in \Cref{tab:simulatedParameters}. Moreover, only six measurements were randomly selected to evaluate the gradient at each iteration.
We define the signal-to-noise ratio (SNR) as
\begin{equation}
\mathrm{SNR}(\boldsymbol{\eta}^*,\boldsymbol{\eta}_\mathrm{true}) = 20\log_{10} \frac{\|\boldsymbol{\eta}_\mathrm{true}\|}{\|\boldsymbol{\eta}_\mathrm{true}-\boldsymbol{\eta}^*\|}\mathrm{dB},
\end{equation}
where $\boldsymbol{\eta}^*$ is the reconstructed RI. To compare the reconstruction on different discretizations, we computed the SNR on the finest grid ($N=1024^2$) by upsampling the reconstructed sample.

\begin{table}
\caption{Stepsizes and regularization parameter on the simulated data for $N=256^2$ and $512^2$.}
\begin{center}
	\begin{tabular}{c|c c|c c}
 \hline
 \hline
 \multirow{2}{*}{$N$}& \multicolumn{2}{c|}{LiS}&\multicolumn{2}{c}{MGH}\\
 &$\gamma$&$\tau$&$\gamma$&$\tau$\\
 
 \hline
 $256^2$&$8.5\times 10^{-4}$&$3.5\times 10^{-3}$&$9\times 10^{-4}$&$4.5\times 10^{-3}$\\
 $512^2$&$4.2\times 10^{-4}$&$7.3\times 10^{-3}$&$3.2\times 10^{-4}$&$9.5\times 10^{-3}$\\
 \hline
 \hline
\end{tabular}
\end{center}
\label{tab:simulatedParameters}
\end{table}

We present in \Cref{fig:SimulatedDataRecoSNRCPUTime} the SNR and CPU time
for both LiS and Helmholtz models with different grids.
The fine discretization ($N=512^2$) yields an SNR higher than the coarser grid ($N=256^2$) does, which shows the influence of the discretization in the reconstruction.
Moreover, our visual assessment in~\Cref{fig:SimulatedDataReco} corroborates the quantitative comparison.
For the LiS method, Bi-CGSTAB needs only about twenty iterations to converge because the contrast is mildly hard.
We still observe that the Helmholtz method needs less CPU time than the LiS method for $N=256^2$. 
For $N=512^2$, we see that the Helmholtz method is faster than the LiS method, which illustrates well the advantage of our method for large $N$.

\begin{figure}
\centering
%
%
%
%
	\begin{tikzpicture}[spy using outlines={circle,lens={scale=2}, size=1.5cm, connect spies}]
	\pgfplotsset{every axis legend/.append style={legend pos=south east, anchor=south east,
	legend cell align={left}}}
	\pgfplotsset{grid style={dotted, gray}}
	\begin{groupplot}[enlargelimits=false,group style={group name=mygroup,group size=1 by 1},
	width=0.4*\textwidth,
	height=0.15*\textwidth,
	every axis/.append style={
	title style={anchor=base,yshift=-1mm}, x label style={yshift = 0.5em}, y label style={yshift = .0em}, grid = both },scale only axis,ymax=38]
	\nextgroupplot[xlabel = {Iterations},ylabel={SNR (dB)},title={}]
	\addplot[dashed,very thick,darkblue,line width=1pt] table [search path={figs/Exp2Simulated_1.00_1.16_rad_3.00},x={Iterations}, y=LiS_SNR256GridsBicubic, col sep=comma] {ResultsSimulatedSNR.csv};
	\addplot[dashdotted,very thick,darkblue,line width=1pt] table [search path={figs/Exp2Simulated_1.00_1.16_rad_3.00},x={Iterations}, y=LiS_SNR512GridsBicubic, col sep=comma] {ResultsSimulatedSNR.csv};
	\addplot[dotted,very thick,red,line width=1pt] table [search path={figs/Exp2Simulated_1.00_1.16_rad_3.00},x={Iterations}, y=Helmholtz_SNR256GridsBicubic, col sep=comma] {ResultsSimulatedSNR.csv};
	\addplot[solid,very thick,red,line width=1pt] table [search path={figs/Exp2Simulated_1.00_1.16_rad_3.00},x={Iterations}, y=Helmholtz_SNR512GridsBicubic, col sep=comma] {ResultsSimulatedSNR.csv};
	\legend{LiS ($N=256^2$),LiS ($N=512^2$),MGH ($N=256^2$),MGH ($N=512^2$)}
	\end{groupplot}
	
	\end{tikzpicture}
	
	\begin{tikzpicture}
  \begin{axis}[
    xbar,
    width=0.475*\textwidth,
	height=0.2*\textwidth, enlarge y limits=0.55,
    xlabel={CPU Time (seconds)}, xmin=0, xmax = 11500,
    symbolic y coords={1},
    legend style={at={(0.5,-0.48)},
        anchor=north,legend columns=2},
    ytick=data,yticklabels={,,},
    nodes near coords, nodes near coords align={horizontal},every node near coord/.append style={
    /pgf/number format/fixed zerofill,
    /pgf/number format/precision=0
}
    ]
    \addplot coordinates {(2556.388156078,1)}; 
    \addplot coordinates {(1883.656180455,1)}; 
    
    \addplot coordinates {(9701.511226183,1)}; 
    \addplot coordinates {(6895.363631848,1)}; 
    
    \legend{LiS ($N=256^2$),MGH ($N=256^2$),LiS ($N=512^2$),MGH ($N=512^2$)}
  \end{axis}
\end{tikzpicture}
	\caption{SNR (top) and CPU time (bottom) with the simulated data for $N=256^2$ and $512^2$.}
	\label{fig:SimulatedDataRecoSNRCPUTime}
\end{figure}
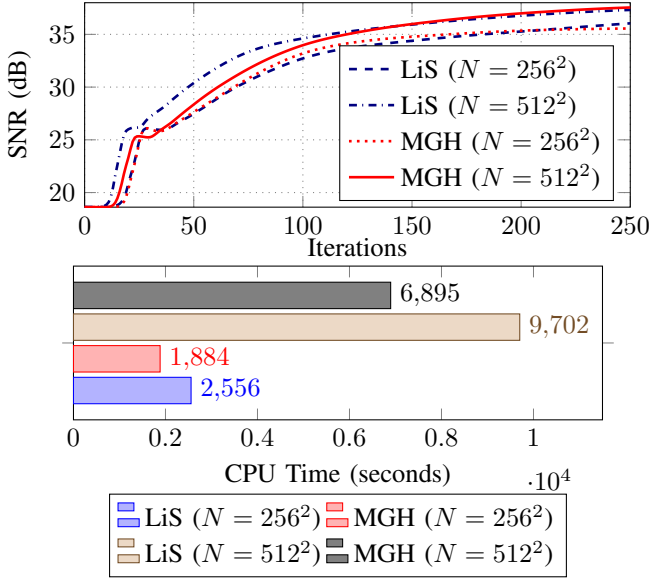

\begin{figure}
	\centering
	\subfigure[LiS: $N=256^2$.]{\includegraphics[scale=0.33]{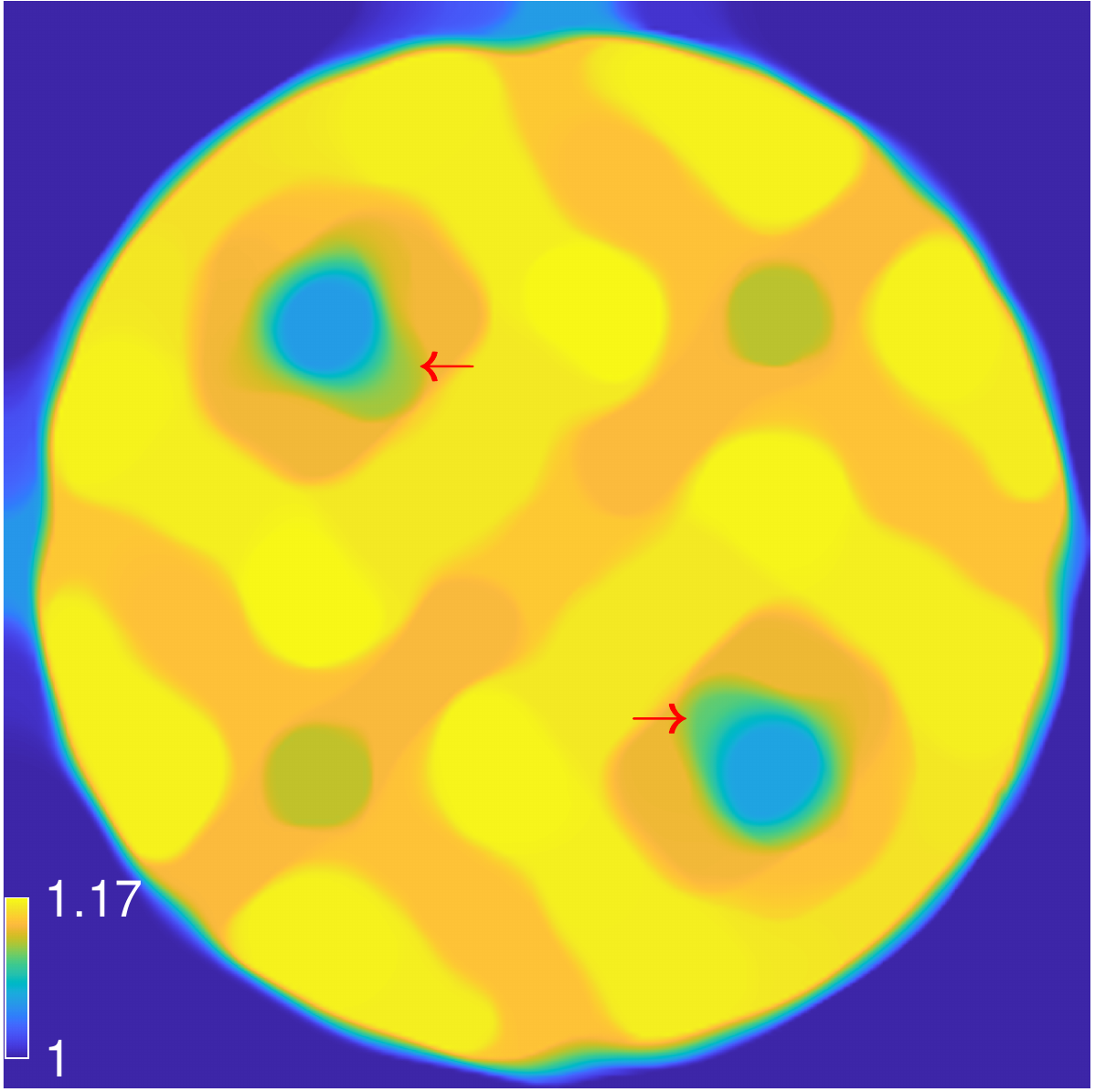}}
	\subfigure[MGH: $N=256^2$.]{\includegraphics[scale=0.33]{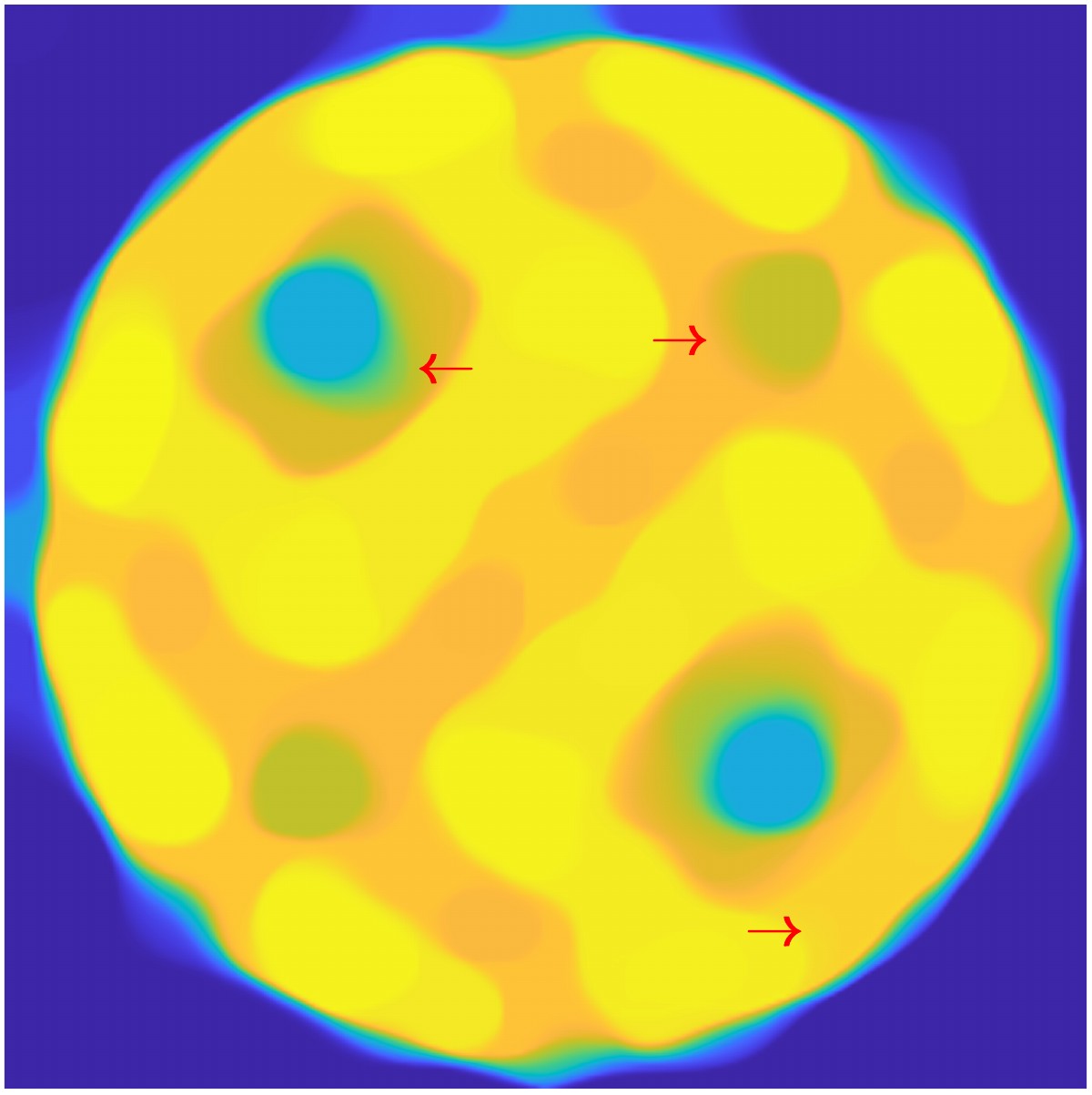}}
	
	\subfigure[LiS: $N=512^2$.]{\includegraphics[scale=0.33]{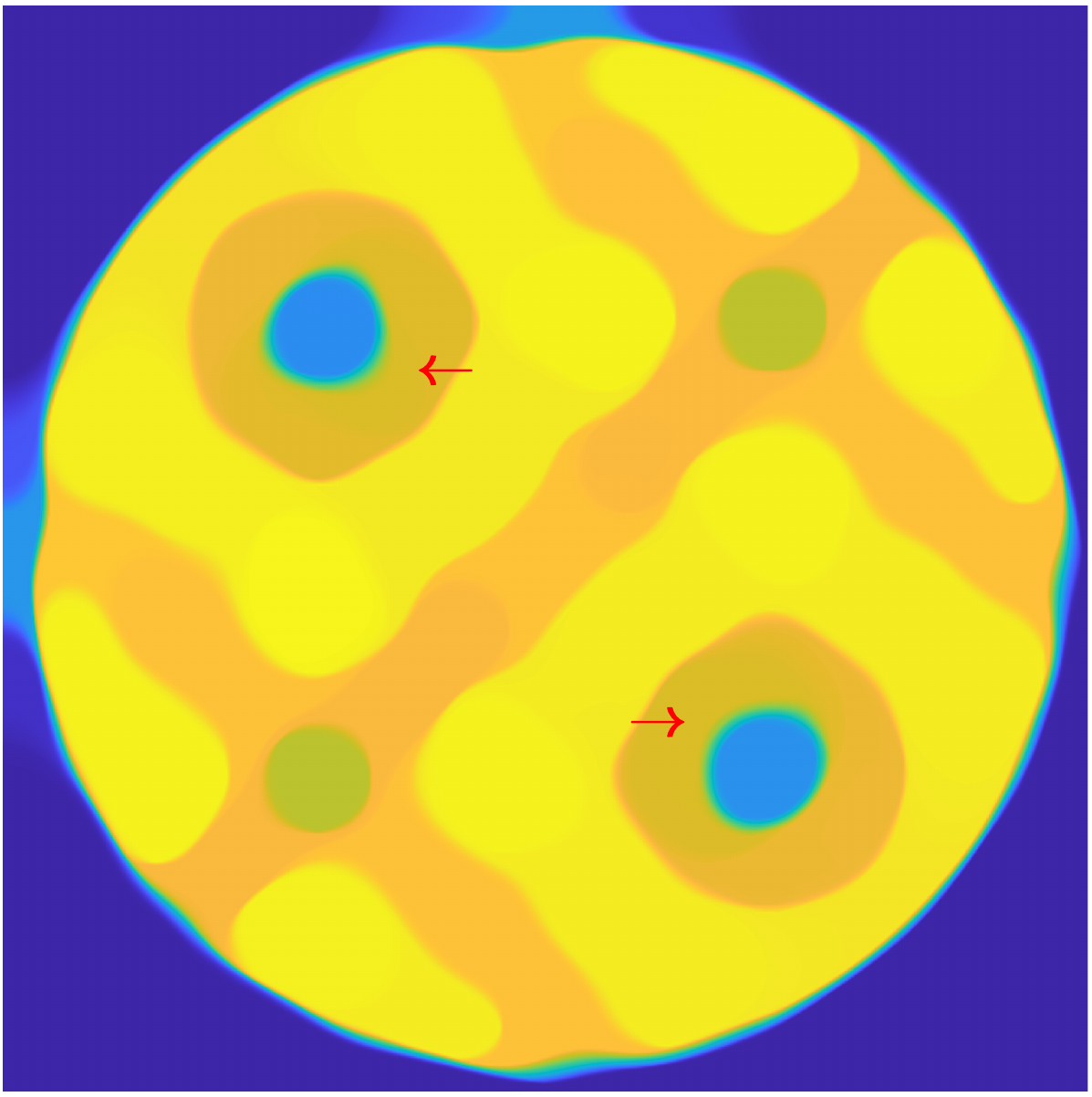}}
	\subfigure[MGH: $N=512^2$.]{\includegraphics[scale=0.33]{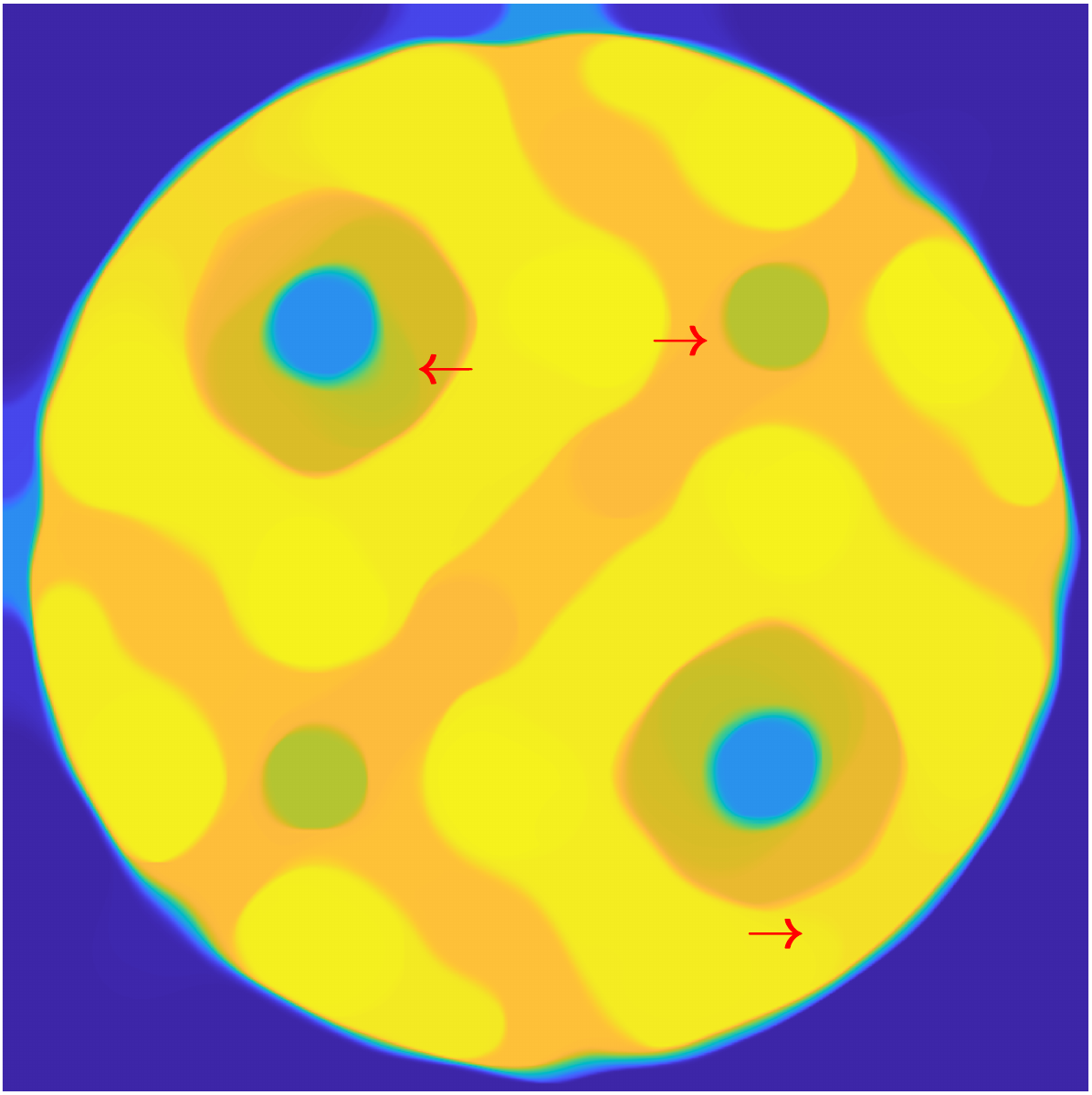}}
	\caption{Reconstructed RIs of the object on $N=256^2$ and $512^2$ grids. All targets are upsampled to a $(1024\times1024)$ grid.}
	\label{fig:SimulatedDataReco}
\end{figure}

\subsection{Inverse Scattering with Experimental Data}
\label{sec:Exps:sub:ExpData}
Now, we study the performance of the Helmholtz model to recover the RIs of three real targets (namely \emph{FoamDielExtTM}, \emph{FoamDielintTM}, and \emph{FoamTwinDielTM}) from the public database provided by the Fresnel Institute \cite{geffrin2005free}.
The samples are fully enclosed in a square domain of length $15$cm.
We discretized the domain over a $(256\times 256)$ grid in our reconstruction.
The sensors were placed circularly around the object at a distance of $1.67$m from its center with a total of $360$ sensors.
Eight (eighteen, respectively) sources  for  \emph{FoamDielExtTM} and \emph{FoamDielintTM} (\emph{FoamTwinDielTM}, respectively) were put uniformly around the object and activated sequentially.
For each activated source, only the $241$ farthest sensors were activated.
In total, $(8\times 241)$ ($(18\times241)$, respectively) measurements for the \emph{FoamDielExtTM} and \emph{FoamDielintTM} (\emph{FoamTwinDielTM}, respectively) targets were obtained.
We used four frequencies of illumination ($3,5,6,8$GHz) to reconstruct the samples, resulting in a total of $(4\times8\times 241)$ measurements~($(4\times18\times 241)$ measurements for \emph{FoamTwinDielTM}).
The expected RIs of the three samples are presented in \Cref{fig:GroundTruthThreeTargets} as reference. 

\begin{figure}
    \centering
   \subfigure[\emph{FoamDielExtTM}]{\includegraphics[scale=0.33]{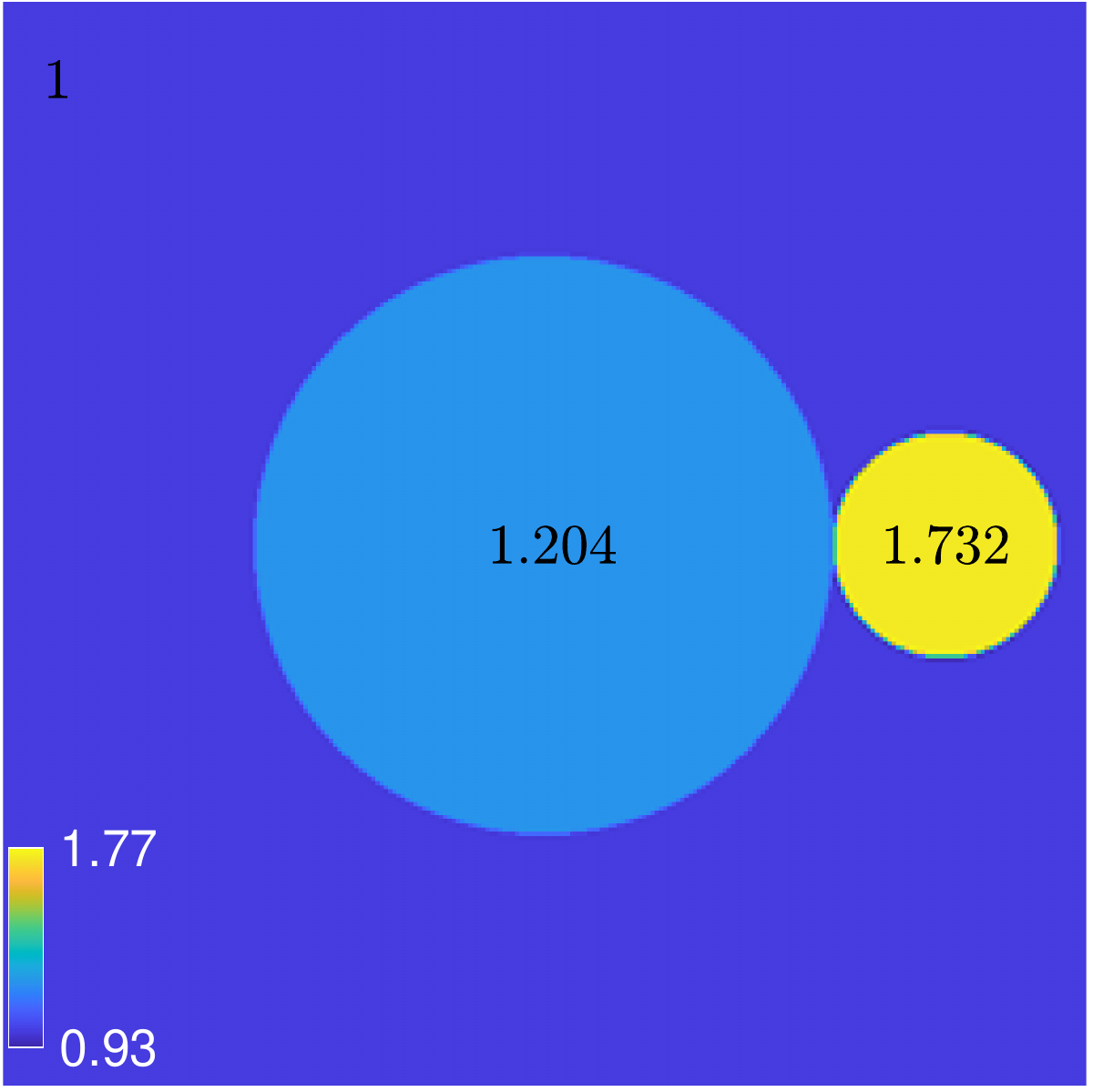}}
   \subfigure[\emph{FoamDielintTM}]{\includegraphics[scale=0.33]{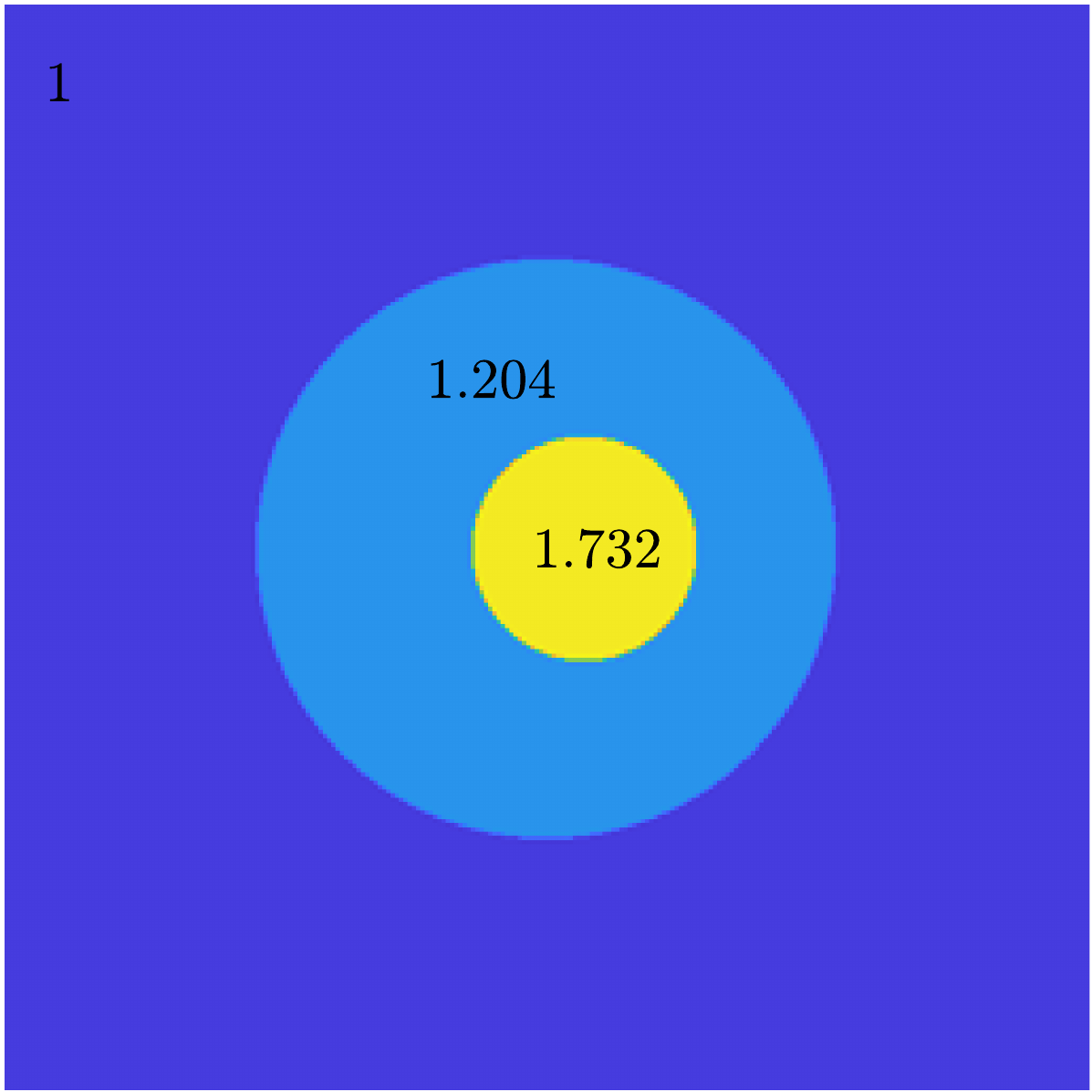}}
   
   \subfigure[\emph{FoamTwinDielTM}]{\includegraphics[scale=0.33]{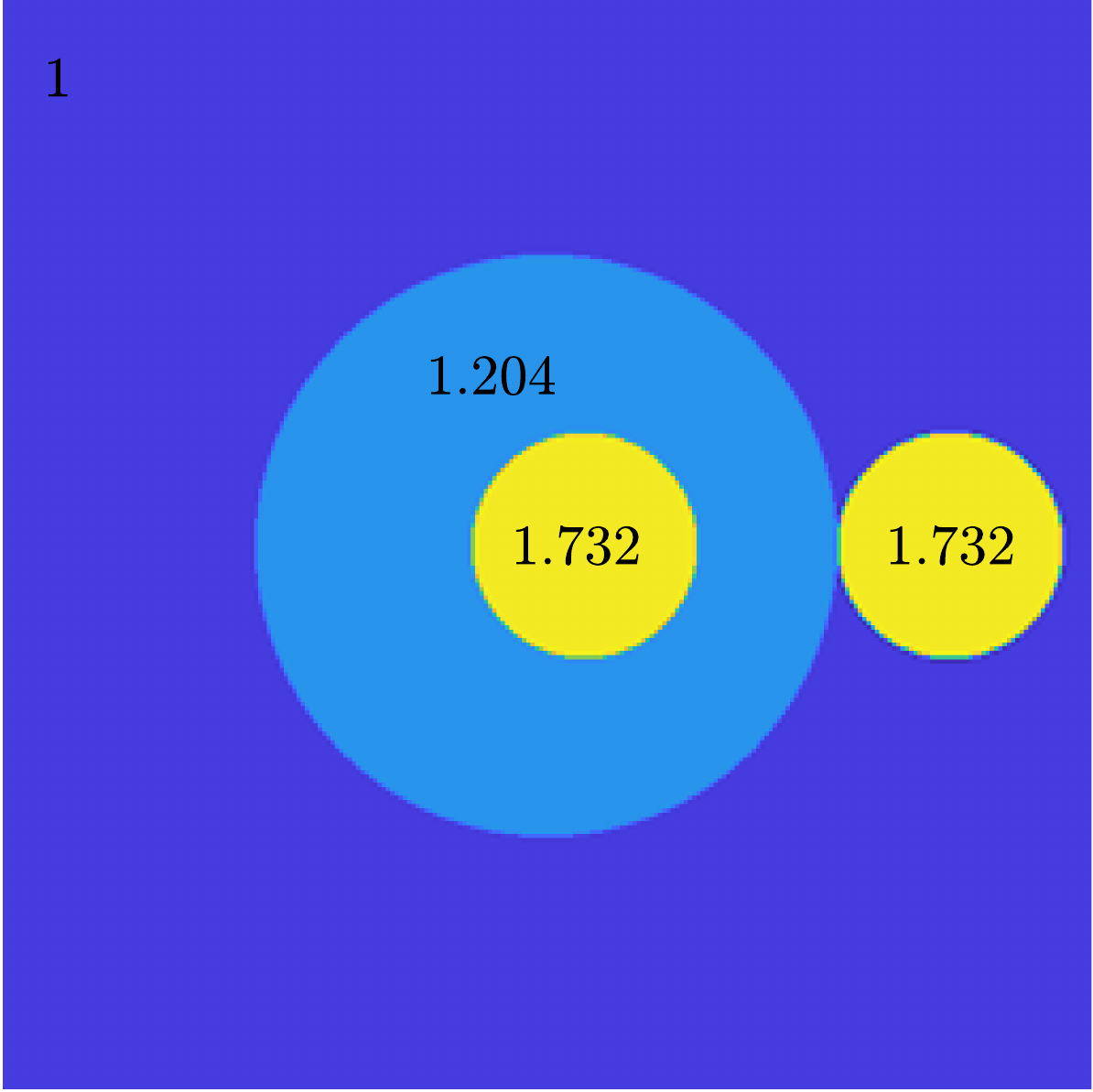}}
    \caption{RIs of three real targets in the Fresnel database.}
    \label{fig:GroundTruthThreeTargets}
\end{figure}

For the reconstruction, we randomly selected a fourth of the measurements to evaluate the gradient at each iteration and performed $150$ iterations.
The stepsize and regularization parameter are summarized in \Cref{tab:experimentalParameters}. From \Cref{fig:RealDataRecov_FoamDielExtTM,fig:RealDataRecov_FoamDielintTM,fig:RealDataRecov_FoamTwinDielTM}, we see that both the LiS and Helmholtz models successfully recover the RIs of real targets with similar performance. Moreover, we observe that the Helmholtz model with the proposed MG solver is faster than the LiS model for all three targets, thus demonstrating the efficiency of our method.

\begin{table}
\caption{Stepsizes and regularization parameter on the experimental data for $N=256^2$.}
\begin{center}
\begin{adjustbox}{scale=0.95}
	\begin{tabular}{c|c c|c c}
 \hline
 \hline
 \multirow{2}{*}{\scriptsize Target}& \multicolumn{2}{c|}{LiS}&\multicolumn{2}{c}{MGH}\\
 &$\gamma$&$\tau$&$\gamma$&$\tau$\\
 \hline
 \emph{\scriptsize FoamDielExtTM}&$4\times 10^{-4}$&$9\times 10^{-3}$&$1.1\times 10^{-3}$&$8.1\times 10^{-3}$\\
 \emph{\scriptsize FoamDielintTM}&$4\times 10^{-4}$&$1.9\times 10^{-2}$&$1\times 10^{-3}$&$7\times 10^{-3}$\\
 \emph{\scriptsize FoamTwinDielTM}&$3\times 10^{-4}$&$1\times 10^{-2}$&$7\times 10^{-4}$&$7.5\times 10^{-3}$\\
 \hline
 \hline
\end{tabular}
\end{adjustbox}
\end{center}
\label{tab:experimentalParameters}
\end{table}

\begin{figure}
	\centering
	\subfigure[LiS SNR: $26.95$dB.]{\includegraphics[scale=0.33]{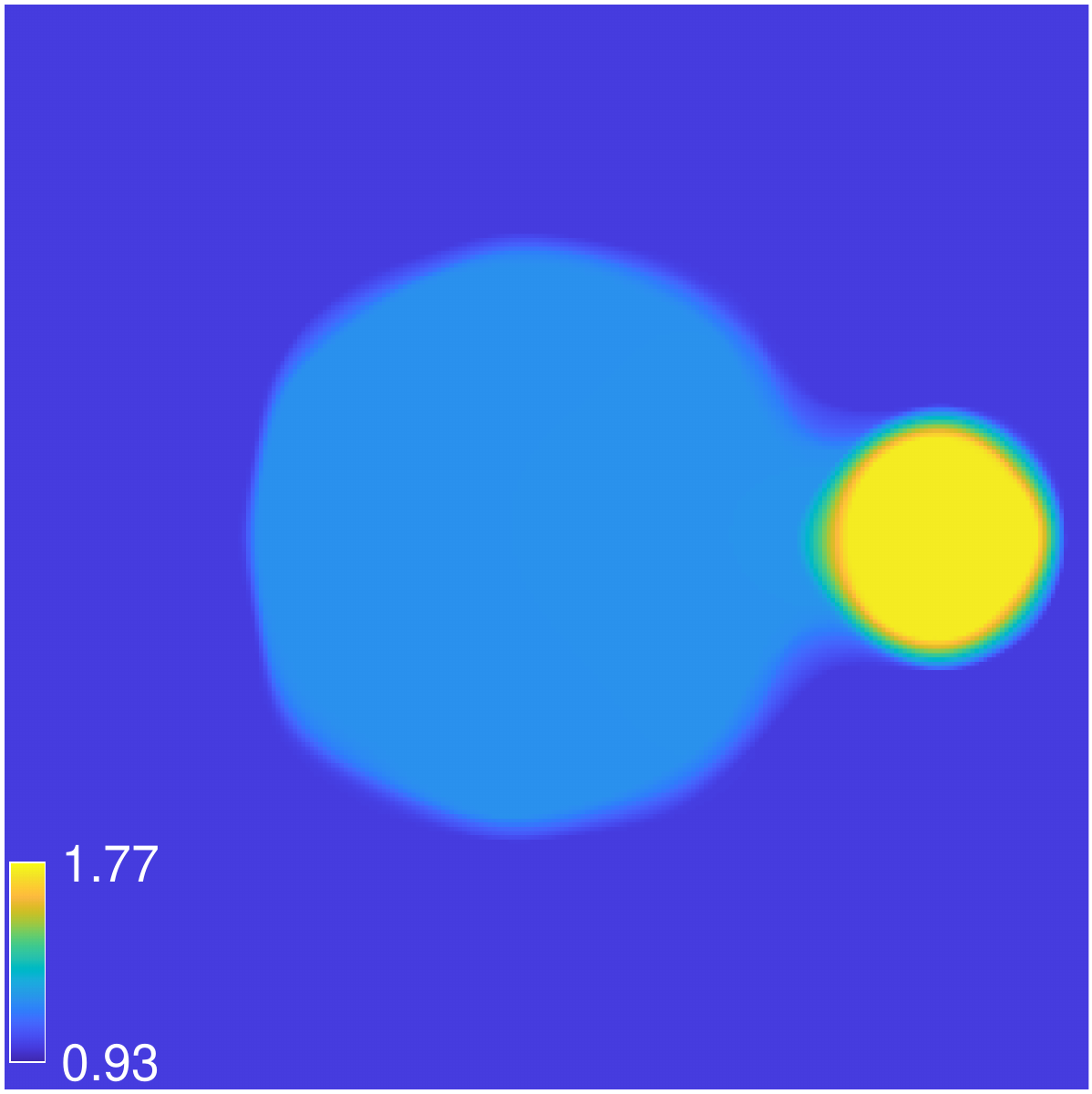}}
	\subfigure[MGH SNR: $26.72$dB.]{\includegraphics[scale=0.33]{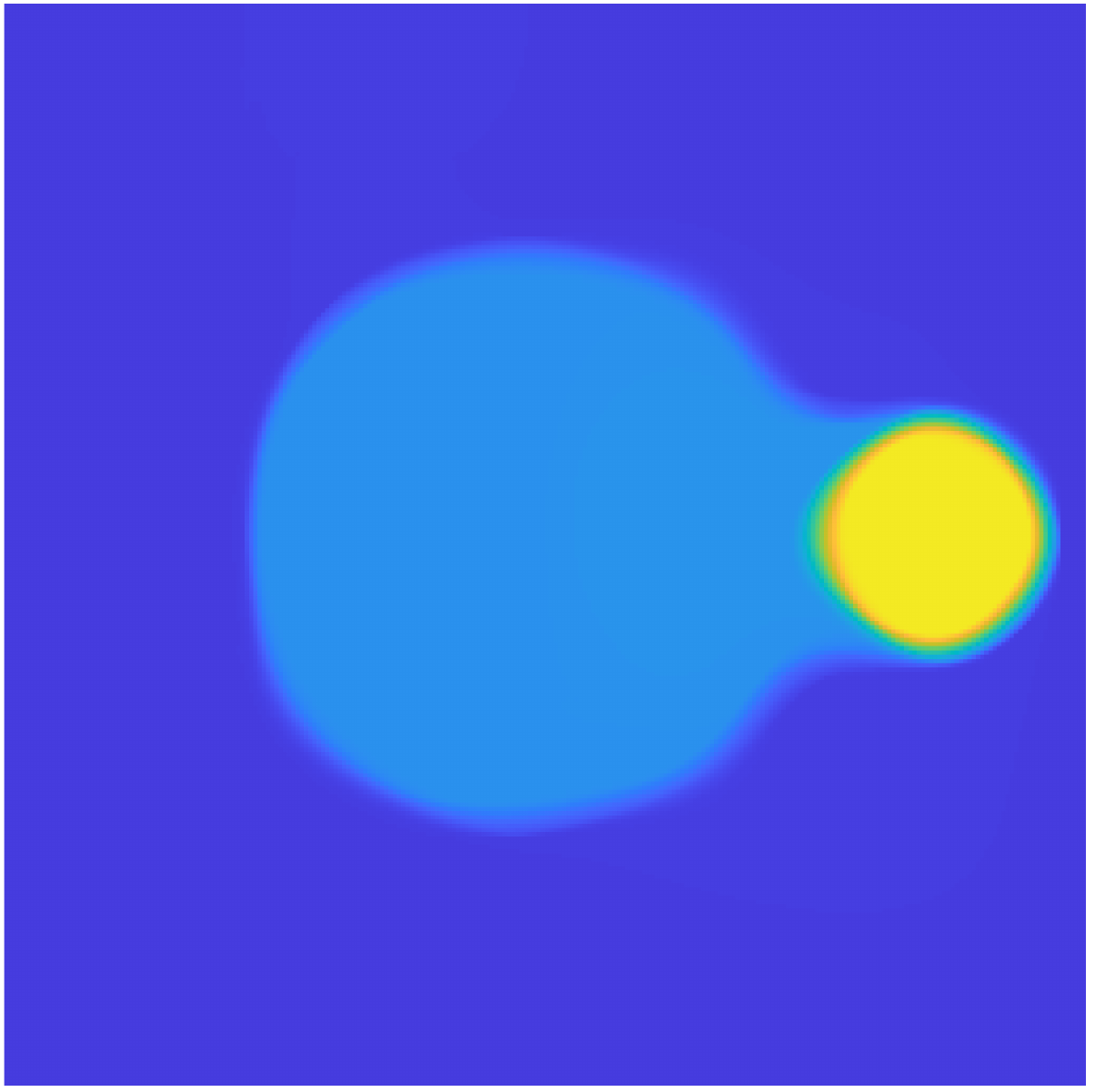}}\\
	\begin{tikzpicture}
 	\pgfplotsset{every axis legend/.append style={legend pos=south east,anchor=south east,font=\normalsize, legend cell align={left}}}
     \pgfplotsset{grid style={dotted, gray}}
 	\begin{groupplot}[enlargelimits=false,scale only axis, group style={group size=1 by 1,x descriptions at=edge bottom,group name=mygroup},
 	width=0.42*\textwidth,
 	height=0.1*\textwidth,
 	every axis/.append style={font=\normalsize,title style={anchor=base,yshift=-1mm}, x label style={yshift = 0.5em}, y label style={yshift = -.5em}, grid = both},xlabel = {Iterations}, ymax = 27,
 	]
	
     \nextgroupplot[ylabel={SNR (dB)}]
     \addplot[dashed,very thick,darkblue,line width=1pt] table [search path={figs/ExpCResultsMultiPlot},x expr=\coordindex, y=LiSSNR, col sep=comma] {ReconstructionFoamDielExtTMSNR.csv};
     \addplot[solid,very thick,red,line width=1pt] table [search path={figs/ExpCResultsMultiPlot},x expr=\coordindex, y=HelmholtzSNR, col sep=comma] {ReconstructionFoamDielExtTMSNR.csv};
     \legend{LiS,MGH};
     \end{groupplot}%
    \end{tikzpicture}
    \begin{tikzpicture}
  \begin{axis}[
    xbar,
    width=0.5*\textwidth,
	height=0.15*\textwidth,
	enlarge y limits=0.5,
    xlabel={CPU Time (seconds)},
    symbolic y coords={1},
    legend style={at={(0.5,-0.85)},
        anchor=north, legend columns=2},yticklabels={,,},
    ytick=data,
    xmin = 0,xmax = 790,
    nodes near coords, nodes near coords align={horizontal},
    every node near coord/.append style={
        /pgf/number format/fixed zerofill,
        /pgf/number format/precision=0
    }
    ]
    \addplot coordinates {(712.628010174,1)};
    \addplot coordinates {(436.8639624,1)};
    \legend{LiS ,MGH}
  \end{axis}
\end{tikzpicture}
		
\caption{Reconstruction of the LiS and Helmholtz models for the \emph{FoamDielExtTM} target.}
	\label{fig:RealDataRecov_FoamDielExtTM}
\end{figure}

\begin{figure}
	\centering
	\subfigure[LiS SNR: $27.99$dB.]{\includegraphics[scale=0.33]{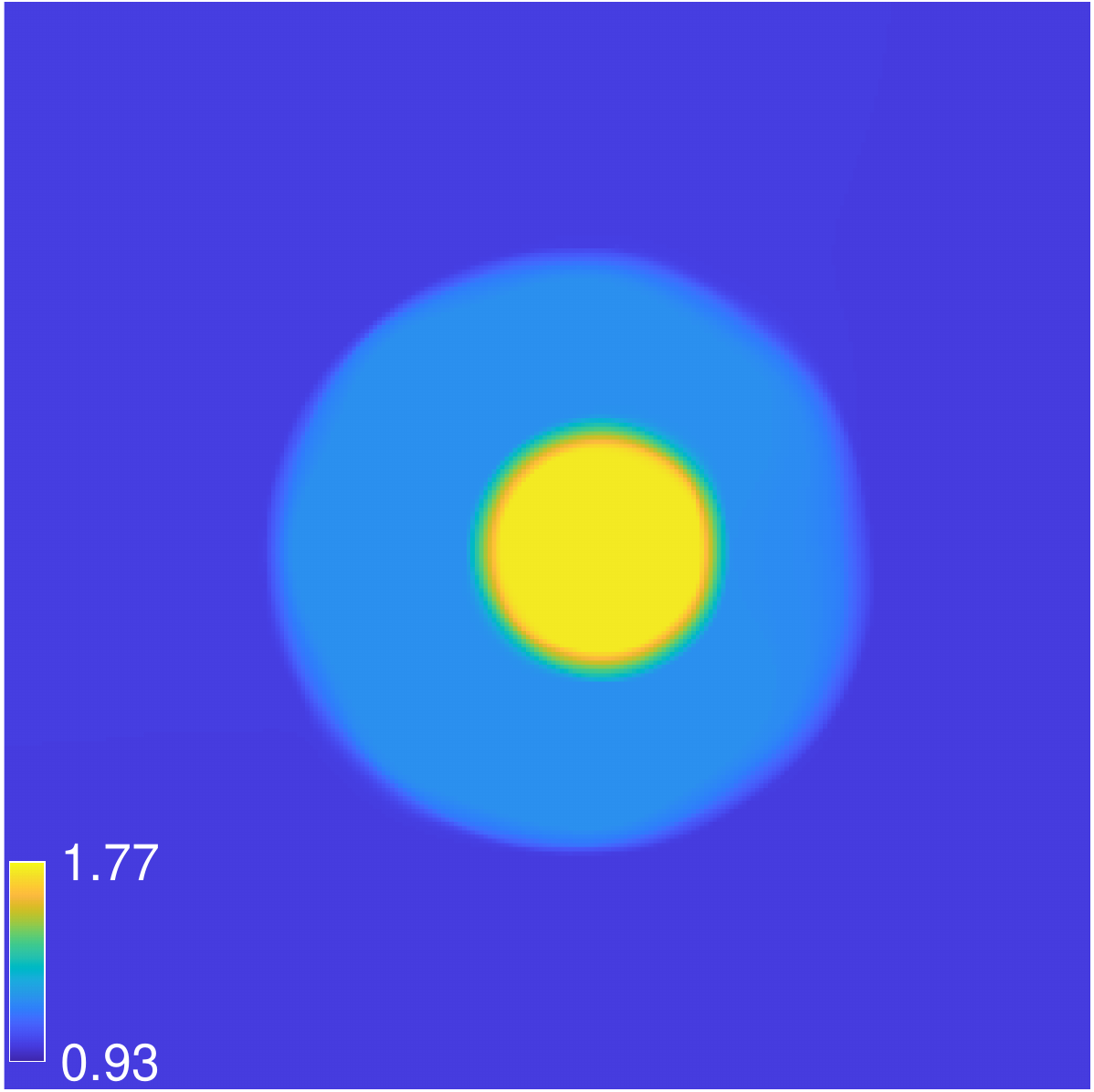}}
	\subfigure[MGH SNR: $28.04$dB.]{\includegraphics[scale=0.33]{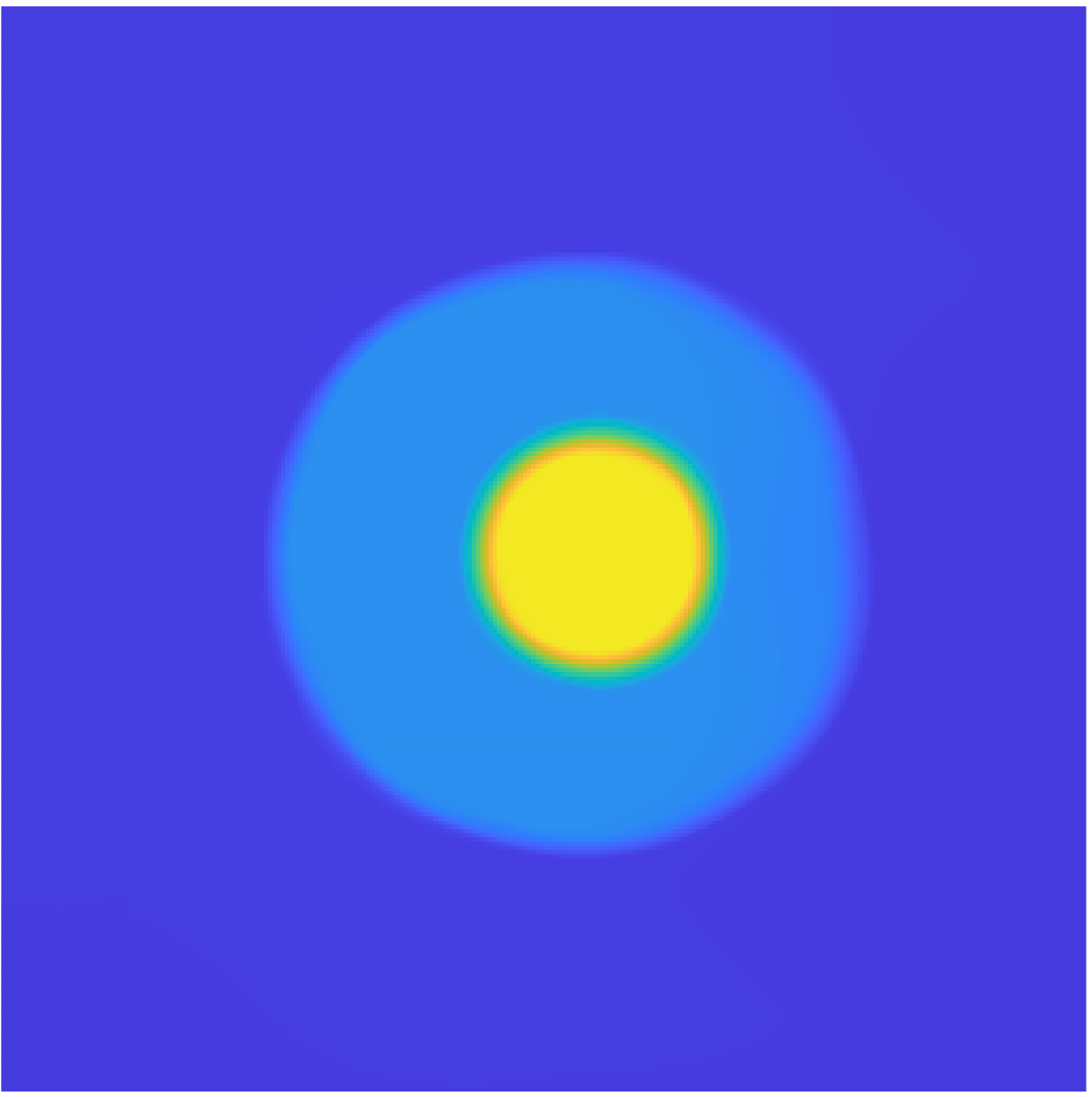}}\\
	\begin{tikzpicture}
 	\pgfplotsset{every axis legend/.append style={legend pos=south east,anchor=south east,font=\normalsize, legend cell align={left}}}
     \pgfplotsset{grid style={dotted, gray}}
 	\begin{groupplot}[enlargelimits=false,scale only axis, group style={group size=1 by 1,x descriptions at=edge bottom,group name=mygroup},
 	width=0.42*\textwidth,
 	height=0.1*\textwidth,
 	every axis/.append style={font=\normalsize,title style={anchor=base,yshift=-1mm}, x label style={yshift = 0.5em}, y label style={yshift = -.5em}, grid = both},xlabel = {Iterations}, ymax = 28.2,
 	]
	
     \nextgroupplot[ylabel={SNR (dB)}]
     \addplot[dashed,very thick,darkblue,line width=1pt] table [search path={figs/ExpCResultsMultiPlot},x expr=\coordindex, y=LiSSNR, col sep=comma] {ReconstructionFoamDielIntTMSNR.csv};
     \addplot[solid,very thick,red,line width=1pt] table [search path={figs/ExpCResultsMultiPlot},x expr=\coordindex, y=HelmholtzSNR, col sep=comma] {ReconstructionFoamDielIntTMSNR.csv};
     \legend{LiS,MGH};
     \end{groupplot}%
    \end{tikzpicture}
    \begin{tikzpicture}
  \begin{axis}[
    xbar,
    width=0.5*\textwidth,
	height=0.15*\textwidth,
	enlarge y limits=0.5,
    xlabel={CPU Time (seconds)},
    symbolic y coords={1},
    legend style={at={(0.5,-0.85)},
        anchor=north, legend columns=2},yticklabels={,,},
    ytick=data,
    xmin = 0,xmax = 790,
    nodes near coords, nodes near coords align={horizontal},
    every node near coord/.append style={
        /pgf/number format/fixed zerofill,
        /pgf/number format/precision=0
    }
    ]
    \addplot coordinates {(701.759519128,1)};
    \addplot coordinates {(432.046426703,1)};
    \legend{LiS ,MGH}
  \end{axis}
\end{tikzpicture}
		
\caption{Reconstruction of the LiS and Helmholtz models for the \emph{FoamDielintTM} target.}
	\label{fig:RealDataRecov_FoamDielintTM}
\end{figure}

\begin{figure}
	\centering
	\subfigure[LiS SNR: $22.14$dB.]{\includegraphics[scale=0.33]{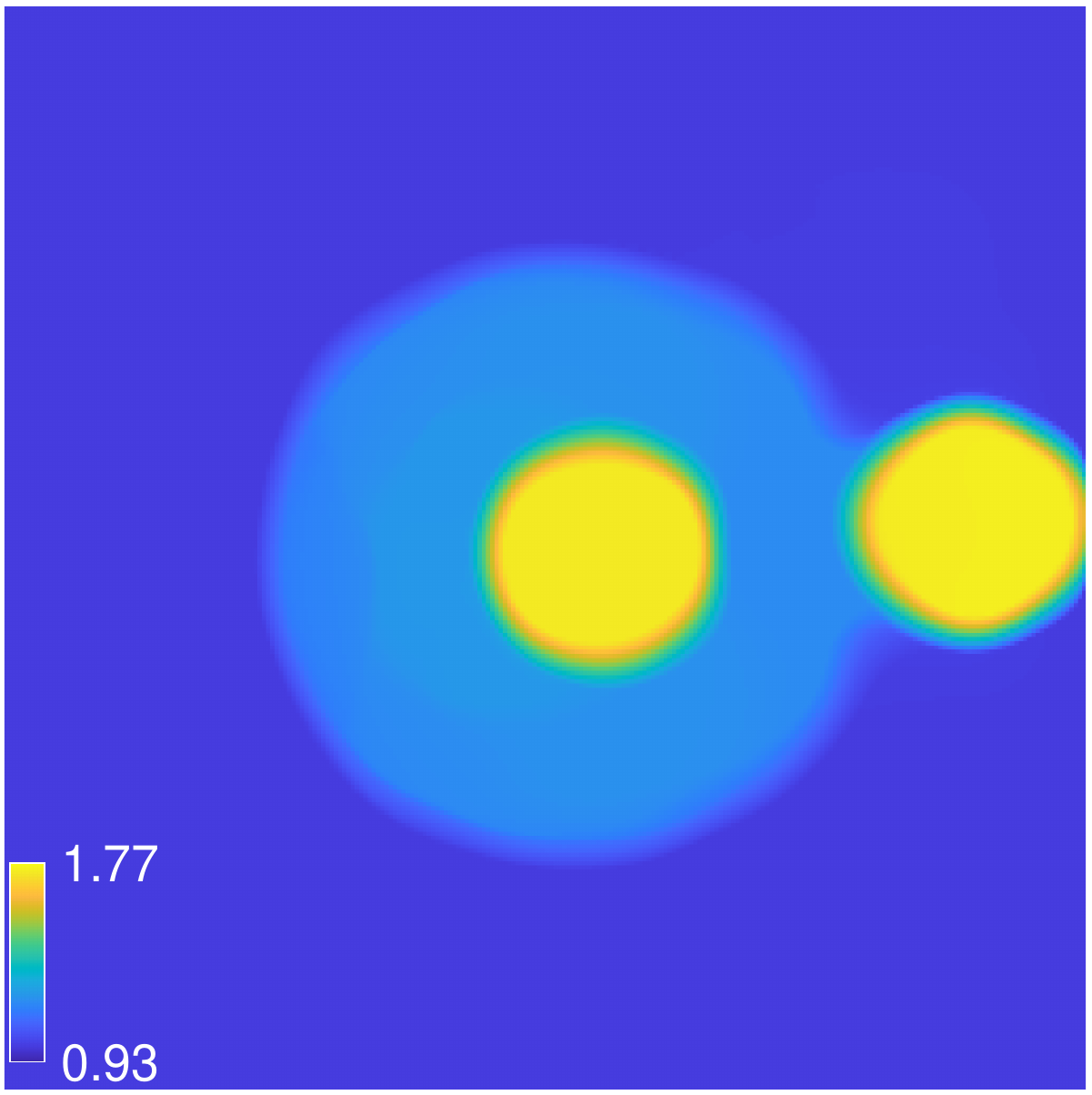}}
	\subfigure[MGH SNR: $22.3$dB.]{\includegraphics[scale=0.33]{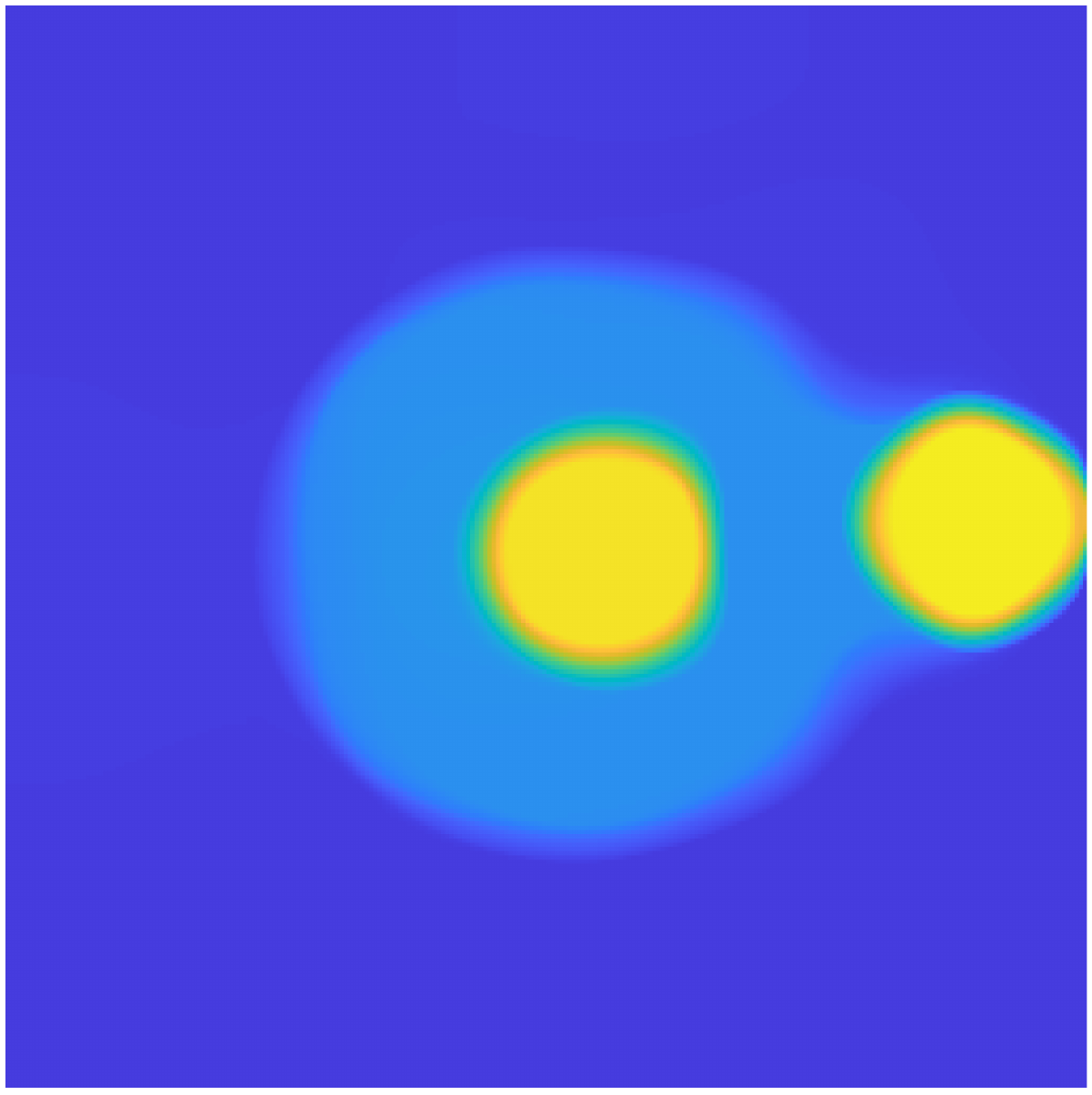}}\\
	\begin{tikzpicture}
 	\pgfplotsset{every axis legend/.append style={legend pos=south east,anchor=south east,font=\normalsize, legend cell align={left}}}
     \pgfplotsset{grid style={dotted, gray}}
 	\begin{groupplot}[enlargelimits=false,scale only axis, group style={group size=1 by 1,x descriptions at=edge bottom,group name=mygroup},
 	width=0.42*\textwidth,
 	height=0.1*\textwidth,
 	every axis/.append style={font=\normalsize,title style={anchor=base,yshift=-1mm}, x label style={yshift = 0.5em}, y label style={yshift = -.5em}, grid = both},xlabel = {Iterations}, ymax = 22.5,
 	]
	
     \nextgroupplot[ylabel={SNR (dB)}]
     \addplot[dashed,very thick,darkblue,line width=1pt] table [search path={figs/ExpCResultsMultiPlot},x expr=\coordindex, y=LiSSNR, col sep=comma] {ReconstructionFoamTwinDielTMSNR.csv};
     \addplot[solid,very thick,red,line width=1pt] table [search path={figs/ExpCResultsMultiPlot},x expr=\coordindex, y=HelmholtzSNR, col sep=comma] {ReconstructionFoamTwinDielTMSNR.csv};
     \legend{LiS,MGH};
     \end{groupplot}%
    \end{tikzpicture}
    \begin{tikzpicture}
  \begin{axis}[
    xbar,
    width=0.5*\textwidth,
	height=0.15*\textwidth,
	enlarge y limits=0.5,
    xlabel={CPU Time (seconds)},
    symbolic y coords={1},
    legend style={at={(0.5,-0.85)},
        anchor=north, legend columns=2},yticklabels={,,},
    ytick=data,
    xmin = 0,xmax = 1780,
    nodes near coords, nodes near coords align={horizontal},
    every node near coord/.append style={
        /pgf/number format/fixed zerofill,
        /pgf/number format/precision=0
    }
    ]
    \addplot coordinates {(1540.258306924,1)};
    \addplot coordinates {(998.597669278,1)};
    \legend{LiS ,MGH}
  \end{axis}
\end{tikzpicture}
		
\caption{Reconstruction of the LiS and Helmholtz models for the \emph{FoamTwinDielTM} target.}
	\label{fig:RealDataRecov_FoamTwinDielTM}
\end{figure}

\section{Conclusions}\label{sec:conclusion}
We have proposed an effective and robust multigrid solver for the Helmholtz equation.
We have shown that our method is adequate and efficient for diffraction tomography,
especially for strongly scattering samples.
This contrasts with Lippmann-Schwinger~(LiS) methods which suffer from slow convergence for such challenging cases.
Moreover, the proposed Jacobian matrix for the Helmholtz model is efficiently computed as well.
For future works, we plan to extend the Helmholtz model to the three-dimensional case, which presents some additional challenges as in the case of LiS.

\ifCLASSOPTIONcaptionsoff
  \newpage
\fi
\bibliographystyle{IEEEtran}
\bibliography{references}

%






\end{document}